\documentclass[journal]{IEEEtran}

\usepackage{amssymb, amsmath, amsthm}
\usepackage{graphicx, subfigure}
\usepackage{url, cite, verbatim}
\usepackage[ruled,vlined,linesnumbered]{algorithm2e}

\usepackage{multirow}
\usepackage{colortbl, color}
\definecolor{kugray5}{RGB}{224,224,224}

\newtheorem{theorem}{Theorem}

\newtheorem{subsec:coding}{subsec:coding}

\newtheorem{lemma}{Lemma}

\newcommand{\ls}[1]  
   {\dimen0=\fontdimen6\the=#1\dimen0
    \advance\lineskip.5\fontdimen5\the\lineskip-\dimen0
    \lineskiplimit=.9\lineskip
    \baselineskip=\lineskip
    \advance\baselineskip\dimen0
    \normallineskip\lineskip
    \normallineskiplimit\lineskiplimit
    \normalbaselineskip\baselineskip
    \ignorespaces
   }

\begin{document}

\title{Adaptive Electricity Scheduling in Microgrids}


\author{Yingsong~Huang~\IEEEmembership{Student~Member,~IEEE},~Shiwen~Mao,~\IEEEmembership{Senior~Member,~IEEE},~and~R.~M.~Nelms~\IEEEmembership{Fellow,~IEEE}%
\thanks{
This work was presented in part at IEEE INFOCOM 2013, Turin, Italy, Apr. 2013~\cite{Huang13}.}
\thanks{Y. Huang, S. Mao, and R.M. Nelms are with the Department of Electrical and Computer Engineering, Auburn University, Auburn, AL 36849-5201. Email: yzh0002@tigermail.auburn.edu, smao@ieee.org, nelmsrm@auburn.edu.}
\thanks{Shiwen Mao is the corresponding author: smao@ieee.org, Tel: (334)844-1845, Fax: (334)844-1809. }
\thanks{Copyright \copyright 2013 by Yingsong Huang, Shiwen Mao, and R. M. Nelms. }
}

\maketitle


\begin{abstract}
Microgrid (MG) is a promising component for future smart grid (SG) deployment. The balance of supply and demand of electric energy is one of the most important requirements of MG management. In this paper, we present a novel framework for smart energy management based on the concept of quality-of-service in electricity (QoSE). Specifically, the resident electricity demand is classified into {\em basic usage} and {\em quality usage}. The basic usage is always guaranteed by the MG, while the quality usage is controlled based on the MG state. The microgrid control center (MGCC) aims to minimize the MG operation cost and maintain the outage probability of quality usage, i.e., QoSE, below a target value, by scheduling electricity among renewable energy resources, energy storage systems, and macrogrid. The problem is formulated as a constrained stochastic programming problem. The Lyapunov optimization technique is then applied to derive an adaptive electricity scheduling algorithm by introducing the QoSE virtual queues and energy storage virtual queues. The proposed algorithm is an online algorithm since it does not require any statistics and future knowledge of the electricity supply, demand and price processes. We derive several ``hard'' performance bounds for the proposed algorithm, and evaluate its performance with trace-driven simulations. The simulation results demonstrate the efficacy of the proposed electricity scheduling algorithm.  
\end{abstract}



\pagestyle{plain}\thispagestyle{plain}

\section{Introduction}

Smart grid (SG) is a modern 
evolution of the utility electricity delivery system. 
SG enhances the traditional power grid through computing, communications, networking, and control technologies throughout the processes of electricity generation, transmission, distribution and consumption. 
The two-way flow of electricity and real-time information is 
a characteristic feature of SG, which offers many technical benefits and flexibilities to both utility providers and consumers, for balancing supply and demand in a timely fashion and improving energy efficiency and 
grid stability.
According to the US 2009 Recovery Act~\cite{RRA09}, 
an SG will replace the traditional system and is expected to save consumer cost and 
reduce America's dependence on foreign oil. 
These goals are to be achieved by improving efficiency and spurring the use of renewable energy resources.

Microgrid (MG) is a promising component for future SG deployment. Due to the increasing deployment of distributed renewable energy resources (DRERs), MG provides a localized cluster of renewable energy generation, storage, distribution and local demand, to achieve reliable and effective energy supply with simplified implementation of SG functionalities~\cite{Fang11,Huang08}. A typical MG architecture is illustrated in Fig.~\ref{fig:microgrid}, consisting of \
DRERs (such as wind turbines and solar photovoltaic cells), energy storage systems (ESS), a communication network (e.g., wireless or powerline communications) for information delivery, an MG central controller (MGCC), and local residents. The MG has centralized control 
with the MGCC~\cite{Huang08}, which exchanges information with local residents, ESS's, and DRERs via the 
information network.  
There is a single common coupling point with the macrogrid. 
When disconnected, the MG works in the {\em islanded mode} and DRERs and ESS's provide electricity to local residents. 
When connected, 
the MG may purchase extra electricity from the macrogrid or sell excess energy back to the market~\cite{Farhangi10}. 

\begin{figure} [!t]
\centering
\includegraphics[width=3.1in]{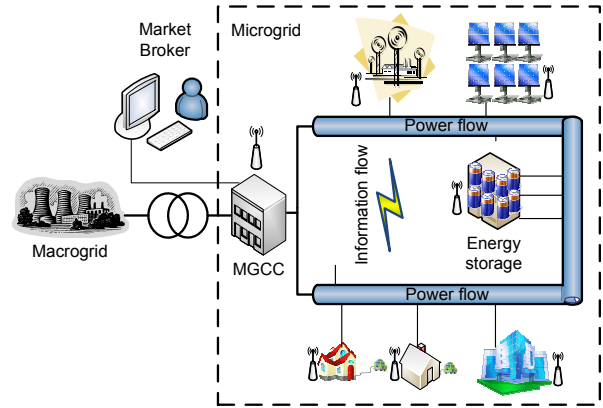}
\caption{Illustrate the microgrid architecture.} 
\label{fig:microgrid}
\end{figure}

The balance of electricity demand and supply is one of the most important requirements in MG management. Instead of matching  supply to demand, smart energy management matches the demand to the available supply using direct load control or off-peak pricing to achieve more efficient capacity utilization~\cite{Fang11}. 
In this paper, we develop a novel control framework for MG energy management, exploiting the two-way flows of electricity and information. 
In particular, we consider two types of electricity usage: 
(i) a pre-agreed {\em basic usage} that is ``hard''-guaranteed, such as basic living usage, and (ii) extra elastic {\em quality usage}
exceeding the pre-agreed level for more comfortable life, such as excessive use of air conditioners or entertainment devices. 
In practice, residents may set their load priority and preference 
to obtain the two types of usage~\cite{Shao11}. The basic usage should be always satisfied, 
while 
the quality usage is controlled by the MGCC according to the grid status, such as DRER generation, ESS storage levels and utility prices. 
The MGCC may \textit{block} some quality usage demand if necessary.  
This can be implemented by incorporating smart meters, smart loads and appliances that can adjust and control their service level through communication flows~\cite{Farhangi10}. 
To quantify residents' satisfaction level,
we define the outage percentage of the quality usage as \textit{Quality of Service in Electricity} (QoSE), which is specified in the service contracts~\cite{Fumagalli04}.  
The MGCC adaptively schedules electricity to keep the QoSE below a target level, and accordingly dynamically balance the load demand to match the available supply. 

In this paper, we investigate the problem of smart energy scheduling by jointly considering renewable energy distribution, ESS management, residential demand management, and utility market participation, aiming to minimize the MG operation cost and guarantee the residents' QoSE. 
The MGCC may serve some quality usage with supplies from the DRERs, ESS's and macrogrid. On the other hand, the MG can also sell excessive electricity back to the macrogrid to compensate for the energy generation cost. The electricity generated from renewable sources is generally random, due to complex weather conditions, while the electricity demand is also random due to the random consumer behavior, and so do the purchasing and selling prices on the utility market. It is challenging to model the random supply, demand, and price processes for MG management, and it may also be costly to have precise, real-time monitoring of the random processes. 
Therefore, a simple, low cost, and optimal electricity scheduling scheme that does not rely on any statistical information of the supply, demand, and price processes would be highly desirable. 

We tackle the MG electricity scheduling problem with a {\em Lyapunov optimization} approach, which is a useful technique to solve stochastic optimization and stability problems~\cite{Tassiulas92}.
We first introduce two virtual queues: QoSE virtual queues and battery virtual queues to transform the QoSE control problem and battery management problem to queue stability problems. Second, we design an adaptive MG electricity scheduling policy based on the Lyapunov optimization method and prove several deterministic (or, ``hard'') performance bounds for the proposed algorithm. The algorithm can be implemented {\em online} because it only relies on the current system status, without needing any future knowledge of the energy demand, supply and price processes. 
The proposed algorithm also converges exponentially due to the nice property of Lyapunov stability design~\cite{Slotine91}. The algorithm is evaluated with trace-driven simulations and is shown to achieve significant efficiency on MG operation cost while guaranteeing the residents' QoSE.

The remainder of this paper is organized as follows. We present the system model and problem formulation in Section~\ref{sec:sys}. 
An adaptive MG electricity scheduling algorithm is designed and analyzed 
in Section~\ref{sec:policy}. Simulation results are presented and discussed in Section~\ref{sec:sim}. We discuss related work in Section~\ref{sec:related}. Section~\ref{sec:con} concludes the paper.


\begin{table} [!t] 
\begin{center}
\caption{Notation}
\label{tab:notation}
\setlength{\extrarowheight}{-3pt}
\begin{tabular}{ll}
       \hline 
       Symbol             & Description \\ 
       \hline
       $N$                & total number of residents \\
       $K$                & total number of batteries \\
       $T$                & total number of slots\\
       $E_k(t)$           & energy level for battery $k$ at time slot $t$ \\
       $R_k(t)$           & recharging energy for battery $k$ at time slot $t$\\
       $D_k(t)$           & discharging energy for battery $k$ at time slot $t$\\
       $E_k^{max}$        & maximum battery energy level for battery $k$ \\
       $E_k^{min}$        & minimum battery energy level for battery $k$\\
       $R_k^{max}$        & maximum supported recharging energy for batter $k$ in a slot\\
       $D_k^{max}$        & maximum supported discharging energy for battery $k$ in a slot\\
       $\lambda_n$        & average quality usage arrival rate for resident $n$ \\
       $\rho_n$           & average outage rate of quality usage for resident $n$ in MG \\
       $\delta_n$           & target QoSE for resident $n$ in MG \\
       $\alpha_n(t)$      & quality usage of residents $n$ in time slot $t$ \\
       $\alpha_n^{max}$     & maximum quality usage of resident $n$ in a single slot \\
       $\alpha_n^b(t)$    & basic electricity usage of resident $n$ in time slot $t$ \\
       $P(t)$             & available electricity from DRERs to supply quality usage in \\
                          &time slot $t$\\
       $U(t)$             & electricity generated from DRERs in time slot $t$ \\
       $Q(t)$            & electricity purchased from macrogrid in time slot $t$ \\
       $S(t)$            & electricity sold on the market in time slot $t$ \\
       $p_n(t)$          & electricity to the resident $n$ \\
       $C(t)$            & purchasing price on the utility market in time slot $t$ \\
       $W(t)$            & selling price ob the utility market in time slot $t$ \\
       $I_n(t)$          & indicator function for outage events of quality usage of\\
                         & resident $n$ in time slot $t$ \\
       $C_{min}$         & minimum purchasing price of utility from macrogrid\\
       $C_{max}$         & maximum purchasing price of utility from macrogrid\\
       $W_{min}$         & minimum selling price of utility to macrogrid\\
       $W_{max}$         & maximum selling price of utility to macrogrid\\
       $X_k(t)$          & battery virtual queue for the battery $k$\\
       $Z_n(t)$          & QoSE virtual queue for the resident $n$\\
       $\vec{\Theta}(t)$ & states of the virtual queues $X_k(t)$ and $Z_n(t)$ \\
       $L(\cdot)$        & Lyapunov function \\
       $\Delta(t)$       & Lyapunov one step drift \\
       $\mathbb{A}(t)$   & proposed scheduling policy including $Q(t), S(t), R_k(t)$,  \\
                         & $D_k(t)$ and $p_n(t)$\\
       $y^*$               & optimal objective value of problem~(\ref{eq:obj})\\
       $\hat{\mathbb{A}}(t)$ & relaxed scheduling policy for problem~\ref{eq:relaxobj}\\
       $\hat{y}$         & optimal objective value of problem~(\ref{eq:relaxobj})\\
       \hline 
\end{tabular}
\end{center}
\end{table}


\section{System Model and Problem Formulation} \label{sec:sys}

\subsection{System Model}

\subsubsection{Overview}

We consider the electricity supply and consumption in an MG as shown in Fig.~\ref{fig:microgrid}. 
We assume that the MG is properly designed 
such that a portion of the electricity demand related to basic living usage (e.g., lighting) 
from the residents, termed {\em basic usage}, can be guaranteed by the minimum capacity of the MG. 
There are randomness in both electricity supply (e.g., weather 
change) and demand (e.g., entertainment 
usage in weekends). To cope with the randomness, 
the MG works in 
the {\em grid-connected} mode and is equipped with ESS's, such as electrochemical battery, 
superconducting magnetic energy storage, flywheel energy storage, etc. The ESS's store excess electricity for future use. 

The MGCC collects information about the resident demands, DRER supplies, and ESS levels through the information network. 
When a resident demand exceeds the pre-agreed level, 
a {\em quality usage} request will be triggered and transmitted to the MGCC. 
The MGCC will then decide the amount of quality usage to be satisfied with energy from the DRERS, the ESS's, or by purchasing electricity from the macrogrid. The MGCC may also decline some quality usage requests. The excess energy can be stored at the ESS's or sold back to the macrogrid for compensating the cost of MG operation.

Without loss of generality, we consider a time-slotted system. 
The time slot duration is determined by the timescale of the demand and supply processes. 

\subsubsection{Energy Storage System Model}

The system model is shown in Fig.~\ref{fig:system}. Consider a battery farm with $K$ independent battery cells, 
which can be recharged and discharged. 
We assume that the batteries are not leaky and do not consider the power loss in recharging and discharging, since the amount is usually small. 
It is easy to relax this assumption by applying a constant percentage on the recharging and discharging processes. For brevity, we also ignore the aging effect of the battery and the maintenance cost, since the cost on the utility market dominates the operation cost of MGs.

\begin{figure} [!t]
\centering
\includegraphics[width=3.3in]{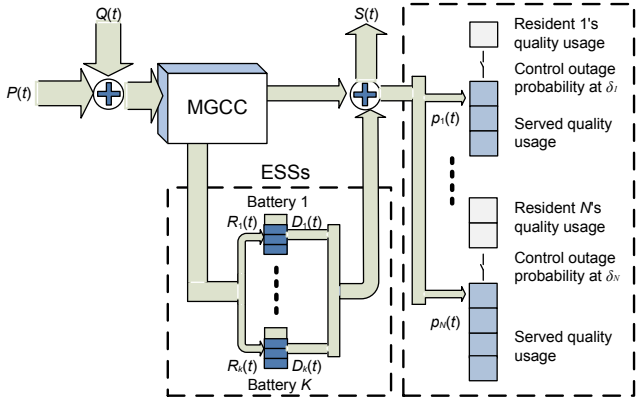}
\caption{The system model considered in this paper.} 
\label{fig:system}
\end{figure}

Let $E_k(t)$ denote the energy level of the the $k$th battery in time slot $t$. 
The capacity of the battery is bounded as 
\begin{equation} \label{eq:batlimit}
E_k^{min} \leq E_k(t) \leq E_k^{max}, \forall \; k, t,
\end{equation}
where $E_k^{max} \geq 0$ is the maximum capacity, and $E_k^{min} \geq 0$ is the minimum energy level required for battery $k$, which may be set by the battery deep discharge protection settings.
The dynamics over time of $E_k(t)$ can be described as
\begin{equation} \label{eq:batdyn}
E_k(t+1) = E_k(t) - D_k(t) + R_k(t), \forall \; k, t,
\end{equation}
where $R_k(t)$ and $D_k(t)$ are the recharging and discharging energy for battery $k$ in time slot $t$, respectively. 
The charging and discharging energy in each time slot are bounded as 
\begin{eqnarray} \label{eq:chadischa}
\left\{ \begin{array}{ll}
   0 \leq R_k(t) \leq R_k^{max}, & \forall \; k, t \\
   0 \leq D_k(t) \leq D_k^{max}, & \forall \; k, t.
   	    \end{array} \right. 
\end{eqnarray}
In each time slot $t$, $R_k(t)$ and $D_k(t)$ are determined such that (\ref{eq:batlimit}) is satisfied in the next time slot. 


Usually the recharging and discharging operations cannot be performed simultaneously, which leads to
\begin{eqnarray} \label{eq:batredecontr}
\left\{ \begin{array}{ll}
  R_k(t) > 0 \Rightarrow D_k(t) = 0, & \forall \; k,t \\
  D_k(t) > 0 \Rightarrow R_k(t) = 0, & \forall \; k,t. 
   	    \end{array} \right. 
\end{eqnarray}



\subsubsection{Energy Supply and Demand Model}

Consider $N$ residents in the MG; each generates basic and quality electricity usage requests, and each can tolerate a prescribed {\em outage probability} $\delta_n$ for the requested quality usage part. 
The MGCC adaptively serves quality usage requests at different levels to maintain the QoSE as well as the stability of the grid. 
The service of quality usage can be different for different residents, depending on individual service agreements. 

Let $\lambda_n$ be the {\em average quality usage arrival rate}, and $\delta_n$ a prescribed outage tolerance (i.e., a percentage) for user $n$. The average {\em outage rate} for the quality usage, $\rho_n$, should satisfy
\begin{equation} \label{eq:qouconst}
  \rho_n \leq \delta_n \cdot \lambda_n. 
\end{equation}
At each time $t$, the quality usage request from resident $n$ is $\alpha_n(t) \in [0, \alpha_n^{max}]$ 
units, which is 
an i.i.d random variable with a general distribution and mean $\lambda_n$. The average rate is 
$\lambda_n = \lim_{t \rightarrow \infty} (1/t) \sum_{\tau=0}^{t-1}{\alpha_n(\tau)}$ according to the Law of Large Numbers.

The DRERs in the MG generate $U(t)$ units of electricity in time slot $t$. $U(t)$ can offer enough capacity to support the pre-agreed {\em basic usage} in the MG, which is guaranteed by islanded mode MG planning.
The electricity is transmitted over power transmission lines. Without loss of generality, we assume the power transmission lines are not subject to outages and the transmission loss is negligible. Let $\alpha_n^b(t)$ be the pre-agreed {\em basic usage} for resident $n$ in time slot $t$, which can be fully satisfied by $U(t)$, i.e.,  $\sum_{n=1}^N\alpha_n^b(t) \leq U(t)$, for all $t$. 
In addition, some {\em quality usage} request $\alpha_n(t)$ may be satisfied if $P(t) = U(t) - \sum_{n=1}^N \alpha_n^b(t) \geq 0$. 
Let $p_n(t)$ be the energy allocated for the quality usage of resident $n$. We have 
\begin{equation} \label{eq:resbal}
  0 \leq p_n(t) \leq \alpha_n(t).
\end{equation}

We define a function $I_n(t) \geq 0$ to indicate the amount of quality usage outage for resident $n$, as 
$I_n(t) =  \alpha_n(t) - p_n(t)$. 
Then the average outage rate can be evaluated as
$\rho_n = \lim_{t \rightarrow \infty} (1/t) \sum_{\tau=0}^{t-1}{I_n(\tau)}$.

The MGCC may purchase additional energy from the macrogrid or sell some excess energy 
back to the macrogrid. Let $Q(t) \in [0, Q_{max}]$
denote the energy purchased from the macrogrid and $S(t) \in [0, S_{max}]$
the energy sold on the market in time slot $t$, 
where $Q_{max}$ and $S_{max}$ are determined by the capacity of the transformers and power transmission lines. Since it is not reasonable to purchase and sell energy on the market at the same time, we have the following constraints
\begin{eqnarray} \label{eq:sellconst}
\left\{ \begin{array}{ll}
	Q(t) > 0 \Rightarrow S(t) = 0, & \forall \; t \\
	S(t) > 0 \Rightarrow Q(t) = 0, & \forall \; t.
	      \end{array} \right. 
\end{eqnarray}
To balance the supply and demand in the MG, we have
\begin{equation} \label{eq:mgbal}
P(t) \hspace{-0.0125in} + \hspace{-0.0125in} Q(t) \hspace{-0.0125in} + \hspace{-0.0125in} \sum_{k=1}^K{D_k(t)} \hspace{-0.0125in} - \hspace{-0.0125in}  S(t) \hspace{-0.0125in} - \hspace{-0.0125in} \sum_{k=1}^K{R_k(t)} \hspace{-0.0125in} = \hspace{-0.0125in} \sum_{n=1}^Np_n(t), \forall \; t.
\end{equation}
\subsubsection{Utility Market Pricing Model}

The price for purchasing electricity from the macrogrid in time slot $t$ is $C(t)$ per unit. The purchasing price depends on the utility market state, such as peak/off time of the day. We assume finite $C(t) \in [C_{min}, C_{max}]$, which is announced by the utility market at the beginning of each time slot and remains constant during the slot period~\cite{Kim11}. Unlike prior work~\cite{Kim11}, we do not require any statistic information of the $C(t)$ process, except that it is independent to the amount of energy to be purchased in that time slot. 

If the MGCC determines to sell electric energy on the utility market, the selling price from the market broker is denoted by $W(t) \in [W_{min}, W_{max}]$ in time slot $t$, which is also a stochastic process with a general distribution. 
We also assume $W(t)$ is known at the beginning of each time slot and independent to the amount of energy to be sold on the market. We assume $C_{max} \geq W_{max}$, $C_{min} \geq W_{min}$ and  $C(t) > S(t)$ for all $t$. That is, the MG cannot make profit by greedily purchasing energy from the market and then sell it back to the market at a higher price simultaneously.

\subsection{Problem Formulation}

Given the above models, a control policy $\mathbb{A}(t) = \{Q(t), S(t), R_k(t), D_k(t), p_n(t)\}$ is designed to minimize the operation cost of the MG and guarantee the QoSE of the residents. We formulate the electricity scheduling problem as
\begin{eqnarray} \label{eq:obj}
\mbox{minimize:} && \hspace{-0.1in} \lim_{t \rightarrow \infty}\frac{1}{t}\sum_{\tau=0}^{t-1}\mathbb{E}\{Q(\tau)C(\tau) - S(\tau)W(\tau)\} \\
\mbox{s.t.} && \hspace{-0.1in} \mbox{(\ref{eq:batlimit}), (\ref{eq:chadischa}), (\ref{eq:batredecontr}), (\ref{eq:qouconst}), (\ref{eq:resbal}), (\ref{eq:sellconst}), (\ref{eq:mgbal}) } \nonumber \\
            && \hspace{-0.1in} \mbox{battery queue stability constraints.} \nonumber
\end{eqnarray}
Problem~(\ref{eq:obj}) is a stochastic programming problem, where the utility prices, generation of DRERs, and consumption of residents are all random. The solution also depends on the evolution of battery states. It is challenging since the supply, demand, and price are all general processes. 


\subsubsection{Virtual Queues} \label{subsubsec:vqs}

We first adopt a {\em battery virtual queue} $X_k(t)$ that tracks the charge level of each battery $k$:
\begin{equation} \label{eq:batteryque}
  X_k(t) = E_k(t) - D_k^{max} - E_k^{min} - VC_{max}, \;\; \forall \; k, t, 
\end{equation}
where $0 < V \leq V_{max} = \min_{k} \left\{\frac{E_k^{max} - E_k^{min}-R_k^{max}-D_k^{max}}{C_{max} - W_{min}} \right\}$ is a constant for the trade-off between system performance and ensuring the battery constraints. 
This constant $V_{max}$ is carefully selected to ensure the evolution of the battery levels always satisfy the battery constraints~(\ref{eq:batlimit}), which will be examined 
in Section~\ref{subsec:perf}.
The virtual queue can be deemed as a shifted version of the battery dynamics in (\ref{eq:batdyn}) as
\begin{equation} \label{eq:virtualbat}
  X_k(t+1) = X_k(t) - D_k(t) + R_k(t), \;\; \forall \; k, t.
\end{equation}
These queues are ``virtual'' because they are maintained by the MGCC control algorithm. Unlike an actual queue, the virtual queue backlog $X_k(t)$ may take negative values. 

We next 
introduce a conceptual {\em QoSE virtual queue} $Z_n(t)$, 
whose dynamics are governed by the system equation as
\begin{equation} \label{eq:vqueue}
  Z_n(t+1) = [Z_n(t) - \delta_n \cdot \alpha_n(t)]^+ + I_n(t), \;\; \forall \; n, t.
\end{equation}
where $[x]^+ = \max\{0, x\}$.


\begin{theorem} \label{propo:vqueue}
If an MGCC control policy stabilizes the QoSE virtual queue $Z_n(t)$, the outage quality usage of resident $n$ will be stabilized at the average QoSE rate $\rho_n \leq {\delta}_n \cdot \lambda_n$.
\end{theorem}
\begin{proof}
According to the system equation~(\ref{eq:vqueue}), we have
\begin{eqnarray} \label{eq:list}
\left\{ \begin{array}{l}
   Z_n(1) \geq Z_n(0) - \delta_n \cdot \alpha_n(0) + I_n(0)  \\
          \cdots  \\
   Z_n(t) \geq Z_n(t-1) -\delta_n \cdot \alpha_n(t-1) + I_n(t-1). 
        \end{array} \right. 
\end{eqnarray}
Summing up the inequalities in (\ref{eq:list}),
we have
\begin{eqnarray} 
Z_n(t) \geq Z_n(0) - \delta_n \cdot \sum_{\tau=0}^{t-1}{\alpha_n(\tau)} + \sum_{\tau=0}^{t-1}I_n(\tau).
\end{eqnarray}
Dividing 
both sides by $t$ 
and letting $t$ go to infinity, we have
\begin{eqnarray}
\lim_{t \rightarrow \infty}\frac{Z_n(t) \hspace{-0.025in} - \hspace{-0.025in} Z_n(0)}{t} \geq \lim_{t \rightarrow \infty}\frac{1}{t} \left[ -\delta_n \sum_{\tau=0}^{t-1}{\alpha_n(\tau)} \hspace{-0.025in} + \hspace{-0.025in} \sum_{\tau=0}^{t-1}{I_n(\tau)} \right]. \nonumber
\end{eqnarray}
Note that $Z_n(0)$ is finite. If $Z_n(t)$ is rate stable by a control policy $I_n(t)$, it is finite for all $t$. We have $\lim_{t \rightarrow \infty}\frac{Z_n(t)-Z_n(0)}{t} = 0$, which yields
$\rho_n \leq \delta_n \cdot \lambda_n$ due to the definitions of $\lambda_n$ 
and $I_n(t)$. 
\end{proof}

\subsubsection{Problem Reformulation}

With Theorem~\ref{propo:vqueue}, we can transform the original problem (\ref{eq:obj}) into a queue stability problem with respect to the QoSE virtual queue and the battery virtual queues, which leads to a system stability design from the control theoretic point of view.  
We have a reformulated stochastic programming problem as follows.
\begin{eqnarray} \label{eq:qobj}
\mbox{minimize:} && \hspace{-0.1in} \lim_{t\rightarrow\infty}\frac{1}{t}\sum_{\tau=0}^{t-1}\mathbb{E}\{Q(\tau)C(\tau) - S(\tau)W(\tau)\} \\
\mbox{s.t.} && \hspace{-0.1in} \mbox{(\ref{eq:chadischa}), (\ref{eq:batredecontr}), (\ref{eq:resbal}), (\ref{eq:sellconst}), (\ref{eq:mgbal}) }\nonumber\\
            && \hspace{-0.1in} \mbox{Battery and QoSE virtual queue stability} \nonumber \\
            && \hspace{-0.1in} \mbox{constraints.} \nonumber
\end{eqnarray}
Theorem~\ref{propo:vqueue} indicates that QoSE provisioning is equivalent to stabilizing the QoSE virtual queue $Z_n(t)$,
while stabilizing the virtual queues~(\ref{eq:virtualbat}) ensures that the battery constraints~(\ref{eq:batlimit}) are satisfied.
We then apply
{\em Lyapunov optimization} to develop an adaptive electricity scheduling policy for problem (\ref{eq:qobj}), in which the policy greedily minimize the Lyapunov drift in every slot $t$ to push the system toward stability.

\subsection{Lyapunov Optimization}


We define the {\em Lyapunov function} for system state $\vec{\Theta}(t) = [\vec{X}(t), \vec{Z}(t)]^T$ with dimension $(N+K)\times1$ as follows, in which $\vec{X}(t) = [X_1(t)\cdots X_K(t)]^T$ and $\vec{Z}(t) = [Z_1(t)\cdots Z_N(t)]^T$. 
\begin{equation}
L(\vec{\Theta}(t)) = \frac{1}{2} \sum_{k=1}^K\left[ X_k(t)\right]^2 + \frac{1}{2} \sum_{n=1}^N\left[ Z_n(t)\right]^2, 
\end{equation}
which is positive definite, since $L(\vec{\Theta}(t)) > 0$ when $\vec{\Theta}(t) \neq \vec{\mathbf{0}}$ 
and $L(\vec{\Theta}(t)) = 0 \Leftrightarrow \vec{\Theta}(t) = \vec{\mathbf{0}}$. 
We then define the conditional one slot {\em Lyapunov drift} as
\begin{equation} \label{eq:driftdef}
\Delta(\vec{\Theta}(t)) = \mathbb{E}\{L(\vec{\Theta}(t+1)) - L(\vec{\Theta}(t))|\vec{\Theta}(t)\}. 
\end{equation} 

With the drift defined as in (\ref{eq:driftdef}), it can be shown that 
\begin{eqnarray}
\Delta(\vec{\Theta}(t)) &=& \frac{1}{2} \mathbb{E} \left\{ \sum_{k=1}^K [(X_k(t+1))^2 - (X_k(t))^2 | X_k(t)]+ \right. \nonumber\\
             & & \hspace*{0.35in} \left. \sum_{n=1}^N [(Z_n(t+1))^2 - (Z_n(t))^2 | Z_n(t)] \right\}\nonumber\\
             &\leq&  B+\sum_{n=1}^N\mathbb{E}\{Z_n(t)(1-\delta_n)\alpha_n(t)|Z_n(t)\} + \nonumber\\
             & & \sum_{k=1}^K\mathbb{E}\{X_k(t)(R_k(t) - D_k(t))|X_k(t)\} - \nonumber \\
             & & \sum_{n=1}^N\mathbb{E}\{(Z_n(t) + \alpha_n(t))p_n(t)|Z_n(t)\}, 
             \label{eq:drift}
\end{eqnarray}
where $B = \frac{1}{2} \sum_{k=1}^K (\max\{D_k^{max}, R_k^{max}\})^2 + \frac{1}{2} \sum_{n=1}^N (2+\delta_n^2)(\alpha_n^{max})^2$ is a constant. 
The derivation of (\ref{eq:drift}) is given in Appendix~\ref{app:drift}.


To minimize the operation cost of the MG, we adopt the {\em drift-plus-penalty method}~\cite{Neely05}. Specifically, we select the control policy $\mathbb{A}(t) = \{Q(t), S(t), R_k(t), D_k(t), p_n(t) \}$ to 
minimize the bound on the drift-plus-penalty as:
\begin{eqnarray} \label{eq:driftpen}
&& \Delta(\vec{\Theta}(t)) + V\mathbb{E}\{Q(t)C(t) - S(t)W(t)|\vec{\Theta}(t)\} \nonumber\\
&\leq& \mbox{right-hand-side of (\ref{eq:drift}) } + \nonumber\\
             & & V\mathbb{E}\{Q(t)C(t) - S(t)W(t)|\vec{\Theta}(t)\},
\end{eqnarray}
where $0 < V \leq V_{max}$ is defined in Section~\ref{subsubsec:vqs} for the trade-off between stability performance 
and operation cost minimization. 
Given the 
current virtual queue states $X_k(t)$ and $Z_n(t)$, market prices $S(t)$ and $W(t)$, available DRERs energy $P(t)$, and the resident quality usage request $\alpha_n(t)$, the optimal policy is the solution to the following problem. 
\begin{eqnarray}\label{eq:objdriftfull}
\mbox{minimize:} && B + \sum_{n=1}^N [Z_n(t)(1-\delta_n)\alpha_n(t)] + \nonumber\\
                && V[Q(t)C(t) - S(t)W(t)] + \nonumber\\
                && \sum_{k=1}^K [X_k(t)(R_k(t) - D_k(t))] - \nonumber \\
                && \sum_{n=1}^N [(Z_n(t) + \alpha_n(t))p_n(t)] \\
    \mbox{s.t.} && \mbox{~(\ref{eq:chadischa}), (\ref{eq:batredecontr}), (\ref{eq:resbal}), (\ref{eq:sellconst}),
                          (\ref{eq:mgbal}) }. \nonumber 
\end{eqnarray}

Since the control policy $\mathbb{A}(t)$ is only applied to the last three terms of~(\ref{eq:objdriftfull}), we can further simplify problem~(\ref{eq:objdriftfull}) as
\begin{eqnarray}\label{eq:objdrift}
\mbox{minimize:} && \hspace{-0.2in} V [Q(t)C(t) - S(t)W(t)] + \sum_{k=1}^K \left[ X_k(t)(R_k(t) - \right. \nonumber\\
                 && \hspace{-0.2in} \left. D_k(t)) \right] - \sum_{n=1}^N [(Z_n(t) + \alpha_n(t))p_n(t)] \\
     \mbox{s.t.} && \hspace{-0.2in} \mbox{~(\ref{eq:chadischa}), (\ref{eq:batredecontr}), (\ref{eq:resbal}), (\ref{eq:sellconst}),
                           (\ref{eq:mgbal}) }, \nonumber 
\end{eqnarray}
which can be solved based on observations of the  
current system state $\{X_k(t), Z_n(t), C(t), W(t), P(t), \alpha_n(t)\}$. 

\section{Optimal Electricity Scheduling} \label{sec:policy}

\subsection{Properties of Optimal Scheduling}

With the Lyapunov penalty-and-drift method, we transform problem~(\ref{eq:qobj}) to problem~(\ref{eq:objdrift}) to be solved for each time slot. The solution only depends on the current system state; there is no need for the statistics of the supply, demand and price processes and no need for any future information. The solution algorithm to this problem is thus an {\em online algorithm}. 
We have the following properties for the optimal scheduling. 
\begin{lemma} \label{lm:lm1}
The optimal solution to problem~(\ref{eq:objdrift}) has the following properties:
\begin{enumerate} \label{enu:enu1}
\item If $Q(t) > 0$, we have $S(t) = 0$,
	\begin{enumerate} 
		\item If $X_k(t) > -VC(t)$, the optimal solution always selects $R_k(t) = 0$;
		if $X_k(t) < -VC(t)$, the optimal	solution always selects $D_k(t) = 0$. \label{item:item1a}
		\item If $Z_n(t) > VC(t) - \alpha_n(t)$, the optimal solution always selects $p_n(t) \geq (1-\delta_n)\alpha_n(t)$; 
		if $Z_n(t) < VC(t) - \alpha_n(t)$, the optimal solution always selects $p_n(t) = 0$. \label{item:item1b}
	\end{enumerate}

\item When $Q(t) = 0$, we have $S(t) > 0$, \label{enu:enu2}
	\begin{enumerate}
		\item If $X_k(t) > -VW(t)$, the optimal solution always selects $R_k(t) = 0$; 
		if $X_k(t) < -VW(t)$, the optimal solution always selects $D_k(t) = 0$.
		\item If $Z_n(t) > VW(t) - \alpha_n(t)$, the optimal solution always selects $p_n(t) \geq (1-\delta_n)\alpha_n(t)$; 
		if $Z_n(t) < VW(t) - \alpha_n(t)$, the optimal solution always selects $p_n(t) = 0$.
	\end{enumerate}
\end{enumerate}
\end{lemma}

The proof of Lemma~\ref{lm:lm1} is given in Appendix~\ref{app:lm1}. 

\begin{lemma}\label{lm:lm2}
The optimal solution to the battery management problem has the following properties:
\begin{enumerate}
\item If $X_k(t) > -VW_{min}$, the optimal solution always selects $R_k(t) = 0$.
\item If $X_k(t) < -VC_{max}$, the optimal solution always selects $D_k(t) = 0$.
\end{enumerate}
\end{lemma}
The proof of Lemma~\ref{lm:lm2} is given in Appendix~\ref{app:lm2}.

\begin{lemma} \label{lm:lm3}
The optimal solution to the QoSE provisioning problem has the following properties:
\begin{enumerate}
\item If $Z_n(t) > VC_{max}$, the optimal solution always selects $p_n(t) \geq (1-\delta_n)\alpha_n(t)$.
\item If $Z_n(t) < VW_{min} - \alpha_{max}$, the optimal solution always selects $p_n(t) = 0$.
\end{enumerate}
\end{lemma}
The proof directly follows Lemma~\ref{lm:lm1} and is similar to the proof of Lemma~\ref{lm:lm2}. We omit the details for brevity. 

Lemma~\ref{lm:lm1} provides useful insights for simplifying the algorithm design, which will be discussed in Section~\ref{subsec:alg}. The intuition behind these lemmas is two-fold. On the ESS management side, if either the purchasing price $C(t)$ or the selling price $W(t)$ is low, the MG prefers to recharge the ESS's to store excess electricity for future use. On the other hand, if either $C(t)$ or $W(t)$ is high, the MG is more likely to discharge the ESS's to reduce the amount of energy to purchase or sell more stored energy back to the macrogrid. On the QoSE provisioning side, if either 
$C(t)$ or 
$W(t)$ is high and the quality usage $\alpha_n(t)$ is low, the MG is apt to decline the quality usage for lower operation cost. 
On the other hand, if either $C(t)$ or $W(t)$ is low and $\alpha_n(t)$ is high, the quality usage are more likely to be granted by purchasing more energy or limiting the sell of energy. 

\subsection{MG Optimal Scheduling Algorithm} \label{subsec:alg}

In this section, we present the MG control policy $\mathbb{A}(t)$ to solve problem~(\ref{eq:objdrift}).
Given the current virtual queue state $\{X_k(t), Z_n(t)\}$, market prices $C(t)$ and $W(t)$, quality usage $\alpha_n(t)$ and available energy $P(t)$ from the DRERS for serving quality usage, 
problem~(\ref{eq:objdrift}) can be decomposed into the following two linear programming (LP) sub-problems 
(since one of $S(t)$ and $Q(t)$ must be zero, see~(\ref{eq:sellconst})). 
\begin{eqnarray} \label{eq:objdriftsub1}
\hspace{-0.1in} \mbox{minimize:} && \hspace{-0.25in} VQ(t)C(t) + \sum_{k=1}^K [X_k(t)(R_k(t) - D_k(t))] -\nonumber\\
                 && \hspace{-0.25in} \sum_{n=1}^N((Z_n(t) + \alpha_n(t))p_n(t)) \\
\hspace{-0.1in} \mbox{s.t.} && \hspace{-0.25in} S(t) = 0, \mbox{(\ref{eq:chadischa}), (\ref{eq:batredecontr}), (\ref{eq:resbal}), (\ref{eq:mgbal}). } \nonumber
\end{eqnarray}
\begin{eqnarray} \label{eq:objdriftsub2}
\hspace{-0.1in} \mbox{minimize:} && \hspace{-0.25in} -VS(t)W(t) + \sum_{k=1}^K [X_k(t)(R_k(t) - D_k(t))] -\nonumber\\
                 && \hspace{-0.25in} \sum_{n=1}^N((Z_n(t) + \alpha_n(t))p_n(t)) \\
\hspace{-0.1in} \mbox{s.t.~} && \hspace{-0.25in} Q(t) = 0, \mbox{(\ref{eq:chadischa}), (\ref{eq:batredecontr}), (\ref{eq:resbal}), (\ref{eq:mgbal}). } \nonumber
\end{eqnarray}

In sub-problem~(\ref{eq:objdriftsub1}), we set $R_k(t) = 0$ if $X_k(t) > -VC(t)$, and $D_k(t) = 0$ if $X_k(t) < -VC(t)$ according to Lemma~\ref{lm:lm1}. Also, if $Z_n(t) < VC(t) - \alpha_n(t)$, we set $p_n(t) = 0$; otherwise, we reset constraint~(\ref{eq:resbal}) to a smaller search space of $(1-\delta_n)\alpha_n(t) \leq p_n(t) \leq \alpha_n(t)$. We take a similar approach for solving 
sub-problem~(\ref{eq:objdriftsub2}) by replacing $C(t)$ with $W(t)$. 
Then we compare the objective values of the two sub-problems and select the more competitive solution as the MG control policy. The complete algorithm is presented in Algorithm~\ref{tab:algorithm}. 

\begin{algorithm} [!t]
\SetAlgoLined
   MGCC initializes the QoSE target to $\delta_n$ and the virtual queues backlogs $Z_n(t)$ and $X_k(t)$, for all $n$ and $k$ \;
   \While{TRUE}{
   		Residents send usage request (with basic and quality usage) to MGCC via the information network \;
   		MGCC solves LPs~(\ref{eq:objdriftsub1}) and (\ref{eq:objdriftsub2}) \;
   		MGCC selects the optimal solution $\mathbb{A}(t)$ comparing the solutions 
   		   to~(\ref{eq:objdriftsub1}) and (\ref{eq:objdriftsub2}) \;
   		MGCC updates the virtual queues $X_k(t)$ and $Z_n(t)$ according to~(\ref{eq:virtualbat}) and~(\ref{eq:vqueue}), 
   		   for all $n$ and $k$ \;
   }
\caption{Adaptive Electricity Scheduling Algorithm}
\label{tab:algorithm}
\end{algorithm}


\subsection{Performance Analysis} \label{subsec:perf}

The proposed scheduling algorithm dynamically balances cost minimization and QoSE provisioning. It only requires current system state information (i.e., as an online algorithm) and requires no statistic information about the random supply, demand, and price processes. The algorithm is also robust to non-i.i.d. and non-ergodic behaviors of the processes~\cite{Neely05}. 

\begin{theorem} \label{propo:battery}
The constraint on the ESS battery level $E_k(t)$, $E_k^{min} \leq E_k(t) \leq E_k^{max}$, is always satisfied for all $k$ and $t$. 
\end{theorem}
\begin{proof}
From the battery virtual queue definition~(\ref{eq:batteryque}), the constraint $E_k^{min} \leq E_k(t) \leq E_k^{max}$ is equivalent to 
$$
  -VC_{max} - D_k^{max} \hspace{-0.015in} \leq \hspace{-0.015in} X_k(t) \hspace{-0.015in} \leq \hspace{-0.015in} E_k^{max} - VC_{max} - D_k^{max} - E_k^{min}. 
$$ 
We assume all the batteries satisfy the battery capacity constraint at the initial time $t = 0$, i.e., $E_k^{min} \leq E_k(0) \leq E_k^{max}$, for all $k$.  Supposing the inequalities hold true for time $t$, we then show the inequalities still hold true for time $t+1$.

First, we show $X_k(t+1) \leq E_k^{max} - VC_{max} - D_k^{max} - E_k^{min}$. If $-VW_{min} < X_k(t) \leq E_k^{max} - VC_{max} - D_k^{max} - E_k^{min}$, then with $X_k(t) > -VW_{min} \Rightarrow R_k(t) = 0$ from Lemma~\ref{lm:lm2}, we have $X_k(t+1) = X_k(t) - D_k(t) \leq X_k(t) \leq E_k^{max} - VC_{max} - D_k^{max} - E_k^{min}$. If $X_k(t) \leq -VW_{min}$, then the largest value is $X_k(t+1)=-VW_{min} + R_k^{max}$. For any $0 < V \leq V_{max}$, we have
\begin{eqnarray}
&&\hspace{-0.1in} E_k^{max} - VC_{max} - D_k^{max} - E_k^{min} \nonumber\\
&\hspace{-0.1in}\geq& \hspace{-0.1in} E_k^{max} - \min_{k} \left\{\frac{E_k^{max} - E_k^{min} - R_k^{max} - D_k^{max}}{C_{max} - W_{min}} \right\} C_{max} \nonumber\\
&&\hspace{-0.1in} - D_k^{max} - E_k^{min} \geq E_k^{min} + R_k^{max} \geq X_k(t+1). \nonumber 
\end{eqnarray}
It follows that $X_k(t+1) \leq E_k^{max} - VC_{max} - D_k^{max} - E_k^{min}$.

Next, we show $X_k(t+1) \geq -VC_{max} - D_k^{max}$. Assuming $-VC_{max} - D_k^{max} \leq X_k(t) \leq -VC_{max}$, then from Lemma~\ref{lm:lm2}, we have $X_k(t) \leq -VC_{max} \Rightarrow D_k(t) = 0$. It follows that
$$
  X_k(t+1) = X_k(t) + R_k(t) \geq X_k(t) \geq -VC_{max} - D_k^{max}. 
$$  
If $X_k(t) \geq -VC_{max}$, following~(\ref{eq:batteryque}), we have
\begin{eqnarray}
X_k(t+1) &=& X_k(t) - D_k(t) + R_k(t) \; \geq \; X_k(t) - D_k^{max} \nonumber\\
         &\geq& -VC_{max} - D_k^{max}. \nonumber
\end{eqnarray}
Therefore, we have $X_k(t+1) \geq -VC_{max} - D_k^{max}$. Thus the inequalities also hold true for time $t+1$. 

It follows that 
$E_k^{min} \leq E_k(t) \leq E_k^{max}$ is satisfied under the optimal scheduling algorithm for all $k$, $t$.  
\end{proof}


\begin{theorem} \label{propo:qoubound}
The worst-case backlogs of the QoSE virtual queue for each resident $n$ is bounded by $Z_n(t) \leq Z_n^{max} = VC_{max} + \alpha_n^{max}$, for all $n$, $t$. Moreover, the worst-case 
average amount of outage of quality usage for resident $n$ in a period $T$ is upper bounded by $Z_n^{max} + T \delta_n \alpha_n^{max}$.
\end{theorem}
\begin{proof}
(i) We first prove the upper bound $Z_n^{max}$. Initially, we have $Z_n(0) = 0 \leq VC_{max} + \alpha_n^{max}$. Assume that in time slot $t$ the backlog of the QoSE virtual queue of resident $n$ satisfies $Z_n(t) \leq Z_n^{max} = VC_{max} + \alpha_n^{max}$. We then check the backlog at time $t+1$ and show the bound still holds true. 

If $Z_n(t) > VC_{max}$, following Lemma~\ref{lm:lm3}, the optimal scheduling for the quality usage of resident $n$ satisfies $p_n(t) \geq (1-\delta_n)\alpha_n(t)$. From the virtual queue dynamics~(\ref{eq:vqueue}), we have
\begin{eqnarray} 
Z_n(t+1) \leq [Z_n(t) - \delta_n\alpha_n(t)]^+ + \delta_n\alpha_n(t). \nonumber
\end{eqnarray} 
If $Z_n(t) \geq \delta_n\alpha_n(t)$, we have $Z_n(t+1) \leq Z_n(t) \leq VC_{max} + \alpha_n^{max}$; otherwise, 
it follows that $Z_n(t+1) \leq \delta_n\alpha_n(t) < VC_{max} + \alpha_n^{max}$.

If $Z_n(t) \leq VC_{max}$, we have $Z_n(t+1) \leq [Z_n(t) - \delta_n\alpha_n(t)]^+ + \alpha_n^{max}$. If $Z_n(t) \geq \delta_n\alpha_n(t)$, we have $Z_n(t+1) \leq Z_n(t) - \delta_n\alpha_n(t) + \alpha_n^{max} \leq VC_{max} + \alpha_n^{max}$;
otherwise, we have $Z_n(t+1) \leq \alpha_n^{max} \leq VC_{max} + \alpha_n^{max}$.

Thus we have $Z_n(t+1) \leq Z_n^{max} = VC_{max} + \alpha_n^{max}$. The proof of the QoSE virtual queue backlog bound is completed.

(ii) Consider an interval $[t_1, t_2]$ with length of $T$. Summing~(\ref{eq:vqueue}) from $t_1$ to $t_2$, we have
$Z_n(t_2+1) \geq Z_n(t_1) -\delta_n \sum_{\tau=t_1}^{t_2}\alpha_n(\tau) + \sum_{\tau=t_1}^{t_2}[\alpha_n(\tau) - p_n(\tau)] \geq \sum_{\tau=t_1}^{t_2}[\alpha_n(\tau)-p_n(\tau)] - T \delta_n \alpha_n^{max}$. 
It follows that 
$$
  \sum_{\tau=t_1}^{t_2}[\alpha_n(\tau)-p_n(\tau)] \leq Z_n^{max} + T\delta_n\alpha_n^{max}.
$$ 
\end{proof}

\begin{theorem} \label{propo:objbound}
The average MG operation cost under the adaptive electricity scheduling algorithm in Algorithm\ref{tab:algorithm}, $\hat{y}$, 
is bounded as $y^* \leq \hat{y} \leq y^* + B^*/V$, where $y^*$ is optimal operating cost and $B^* = B + \sum_{n=1}^NZ_n^{max}(1-\delta_n)\alpha_n^{max}$. 
\end{theorem}
\begin{proof}
From Theorem~\ref{propo:battery}, the battery capacity constraints is met in each time slot with the 
adaptive control policy. Take expectation on~(\ref{eq:batdyn})  
and sum it over the period $[0, t-1]$:
$$
  \mathbb{E}\{E_k(t)\} \hspace{-0.0125in} - \hspace{-0.0125in} \mathbb{E}\{E_k(0)\} = \sum_{\tau = 0}^{t-1}[\mathbb{E}\{R_k(\tau)\} \hspace{-0.0125in} - \hspace{-0.0125in} \mathbb{E}\{D_k(\tau)\}], \; \forall \; k.
$$
Since $E_k^{min} \leq E_k (t) \leq E_k^{max}$, we divide both sides 
by $t$ and 
let $t$ go to infinity, to obtain 
\begin{equation} \label{eq:batavg}
\lim_{t\rightarrow\infty}\frac{1}{t}\sum_{\tau=0}^{t-1}\mathbb{E}\{R_k(\tau)\} = \lim_{t\rightarrow\infty}\frac{1}{t}\sum_{\tau=0}^{t-1}\mathbb{E}\{D_k(\tau)\}, \; \forall \; k.
\end{equation}

Consider the the following relaxed version of problem~(\ref{eq:obj}). 
\begin{eqnarray} \label{eq:relaxobj}
\mbox{minimize:} && \hspace{-0.2in} \lim_{t\rightarrow\infty}\frac{1}{t}\sum_{\tau=0}^{t-1}\mathbb{E}\{Q(\tau)C(\tau) - S(\tau)W(\tau)\} \\
     \mbox{s.t.} && \hspace{-0.2in} \mbox{ (\ref{eq:chadischa}), (\ref{eq:batredecontr}), (\ref{eq:qouconst}), (\ref{eq:resbal}), (\ref{eq:sellconst}), (\ref{eq:mgbal}), and (\ref{eq:batavg})}. \nonumber
\end{eqnarray}
Since the constraints in problem~(\ref{eq:relaxobj}) are relaxed from that in problem~(\ref{eq:obj}), the optimal solution to problem~(\ref{eq:obj}) is 
also feasible for problem~(\ref{eq:relaxobj}). The solution of~(\ref{eq:relaxobj}) does not depend on 
battery energy levels. Let the optimal solution for problem~(\ref{eq:relaxobj}) be $\hat{\mathbb{A}}(t) = \{\hat{Q}(t), \hat{S}(t), \hat{R_k}(t), \hat{D_k}(t), \hat{p}_n(t)\}$ and the corresponding object value is $\hat{y} \leq y^*$. 
According to the properties of optimality of stationary and randomized policies~\cite{Neely08}, the optimal solution $\hat{\mathbb{A}}(t)$ satisfies $\mathbb{E}\{\hat{R}_k(t) - \hat{D}_k(t)\} = 0$ and $\hat{y} = \mathbb{E}\{\hat{Q}(\tau)C(\tau) - \hat{S}(\tau)W(\tau)\}$.

We substitute solution $\hat{\mathbb{A}}(t)$ into the right-hand-side of the drift-and-penalty~(\ref{eq:driftpen}). Since our proposed policy minimizes the right-hand-side of~(\ref{eq:driftpen}), we have
\begin{eqnarray} \label{eq:objbound1}
&&\Delta(\vec{\Theta}(t)) + V\mathbb{E}\{Q(t)C(t) - S(t)W(t)|\vec{\Theta}(t)\} \nonumber\\
&\leq& B + \sum_{n=1}^N\mathbb{E}\{Z_n(t)(1-\delta_n)\alpha_n(t)|Z_n(t)\} + \nonumber\\
       & & \sum_{k=1}^KX_k(t)\mathbb{E}\{\hat{R}_k(t) - \hat{D}_k(t)|X_k(t)\} - \nonumber \\
       & & \sum_{n=1}^N(Z_n(t) + \alpha_n(t))\mathbb{E}\{\hat{p}_n(t)|Z_n(t)\} + \nonumber\\
       & & V\mathbb{E}\{\hat{Q}(t)C(t) - \hat{S}(t)W(t)|\vec{\Theta}(t)\} \nonumber \\
       &\leq& B+\sum_{n=1}^NZ_n^{max}(1-\delta_n)\alpha_n^{max} + V \cdot y^* \nonumber.
\end{eqnarray}
The second inequality is due to $\mathbb{E}\{\hat{R}_k(t) - \hat{D}_k(t)\} = 0$, $0 \leq Z_n(t) \leq Z_n^{max}$, $\alpha_n(t) \geq 0$, $p_n(t) \geq 0$, and $\hat{y} \leq y^*$. Taking expectation and sum up from $0$ to $T-1$, we obtain
\begin{eqnarray}
&&\sum_{t=0}^{T-1}V\mathbb{E}\{Q(t)C(t) - S(t)W(t)\} \nonumber\\
&\leq& T\cdot B^* + T\cdot V \cdot y^* - \mathbb{E}\{L(\vec{\Theta}(T))\} + \mathbb{E}\{L(\vec{\Theta}(0))\} \nonumber \\
&\leq& T\cdot B^* + T\cdot V \cdot y^* + \mathbb{E}\{L(\vec{\Theta}(0))\}. \nonumber
\end{eqnarray}
The second inequality is due to the nonnegative property of Lyapunove functions. Divide both sides by $V\cdot T$ and 
let $T$ go to infinity. Since the initial system state $\vec{\Theta}(0)$ is finite, we have
$\lim_{T\rightarrow\infty}\frac{1}{T}\sum_{t=0}^{T-1}V\mathbb{E}\{Q(t)C(t) - S(t)W(t)\} \leq y^* + \frac{B^*}{V}$.  
\end{proof}

It is worth noting that the choice of $V$ controls the optimality of the proposed algorithm. Specifically, a larger $V$ leads to a tighter optimality gap. However, from the proof of  
Theorem~\ref{propo:battery}, $V$ is limited by $V_{max}$, which ensures the feasibility of the battery constraints. This is actually a similar phenomenon to the so-called {\em performance-congestion trade-off}~\cite{Neely08}. Through the definition of $V_{max}$ (see Section~\ref{subsubsec:vqs}), it can be seen that if we invest more on the individual storage components for a larger ESS capacity, the proposed algorithm can achieve a better performance (i.e., a smaller optimality gap).

It is also worth noting that all the performance bounds of the proposed algorithm are deterministic, which provide ``hard'' guarantees for the performance of the proposed adaptive scheduling policy in every time slot.
Unlike probabilistic approaches, the proposed method provides useful guidelines for the MG design, while guaranteeing the MG operation cost, grid stability, and the usage quality of residents. 
\section{Simulation Study} \label{sec:sim}

We demonstrate the performance of the proposed adaptive MG electricity scheduling algorithm through extensive simulations. We simulated an MG with $500$ residents, where the electricity from DRERs is supplied by a wind turbine plant. We use the renewable energy supply data from the Western Wind Resources Dataset published by the National Renewable Energy Laboratory~\cite{WWR}. The ESS's consists of $100$ PHEV Li-ion battery packs, each of which has a maximum capacity of $16$ kWh and the minimum energy level is $0$. The battery can be fully charged or discharged within $2$ hours~\cite{Peterson10}. 

The residents' pre-agreed power demand is uniformly distributed in [$2$ kW, $25$ kW], and the quality usage power is uniformly distributed in [$0$, $10$ kW]. The MG works in the grid-connected mode and may purchase/sell electricity from/to the macrogrid. The utility prices in the macrogrid are obtained from~\cite{ERCOT} and are time-varying. We assume the sell price by the broker is random and below the purchasing price in each time slot. The time slot duration is $15$ minutes. The MGCC serves a certain level of quality usage according to the adaptive electricity scheduling policy. The QoSE target is set to $\delta_n = 0.07$ for all residents. The control parameter is $V = V_{max}$, unless otherwise specified. 

\subsection{Algorithm Performance} 

We first investigate the average QoSEs and total MG operation cost with default settings for a five-day period.  We use MATLAB LP solver for solving the sub-problems (\ref{eq:objdriftsub1}) and (\ref{eq:objdriftsub2}). 
For better illustration, we only show the QoSEs of three randomly chosen users in Fig.~\ref{fig:vmaxqou}. It can be seen that all the average QoSEs converge to the neighborhood of $0.08$ within $200$ slots, which is close to the MG requested criteria $\delta_n=0.07$. In fact the proposed scheme converges exponentially, due to the inherent exponential convergence property in Lyapunov stability based design~\cite{Slotine91}. 

\begin{figure*}[!t]
	\begin{minipage}[t]{.32\linewidth}
		\centering  
		\includegraphics[width=2.1in, height=1.7in]{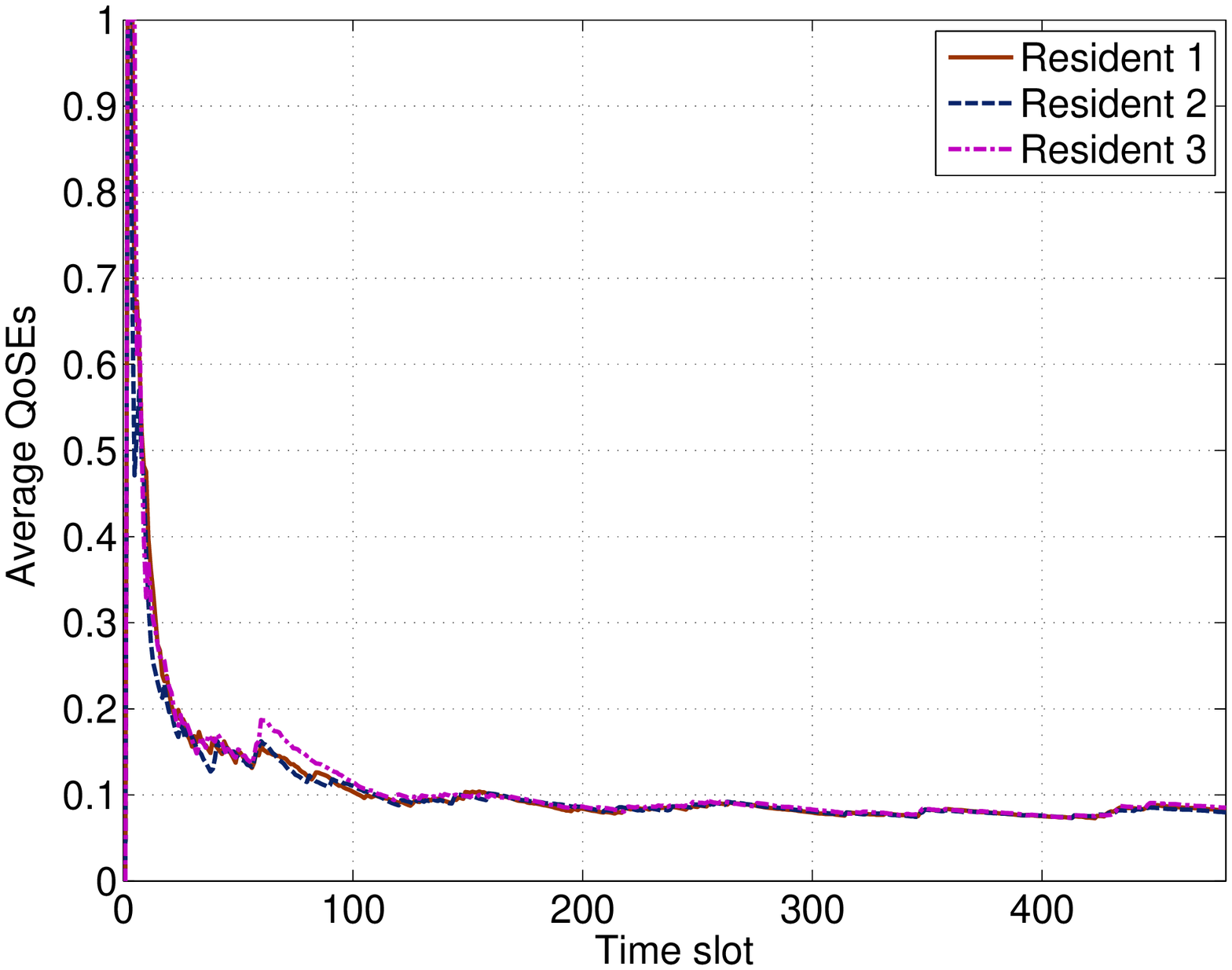}
		\caption{\small Average QoSEs of three residents ($V = V_{max}$).}
		\label{fig:vmaxqou}
		\vspace{-0.1in}
	\end{minipage}
	\hspace{0.05in}
	\begin{minipage}[t]{.32\linewidth}
		\centering  
		\includegraphics[width=2.1in, height=1.7in]{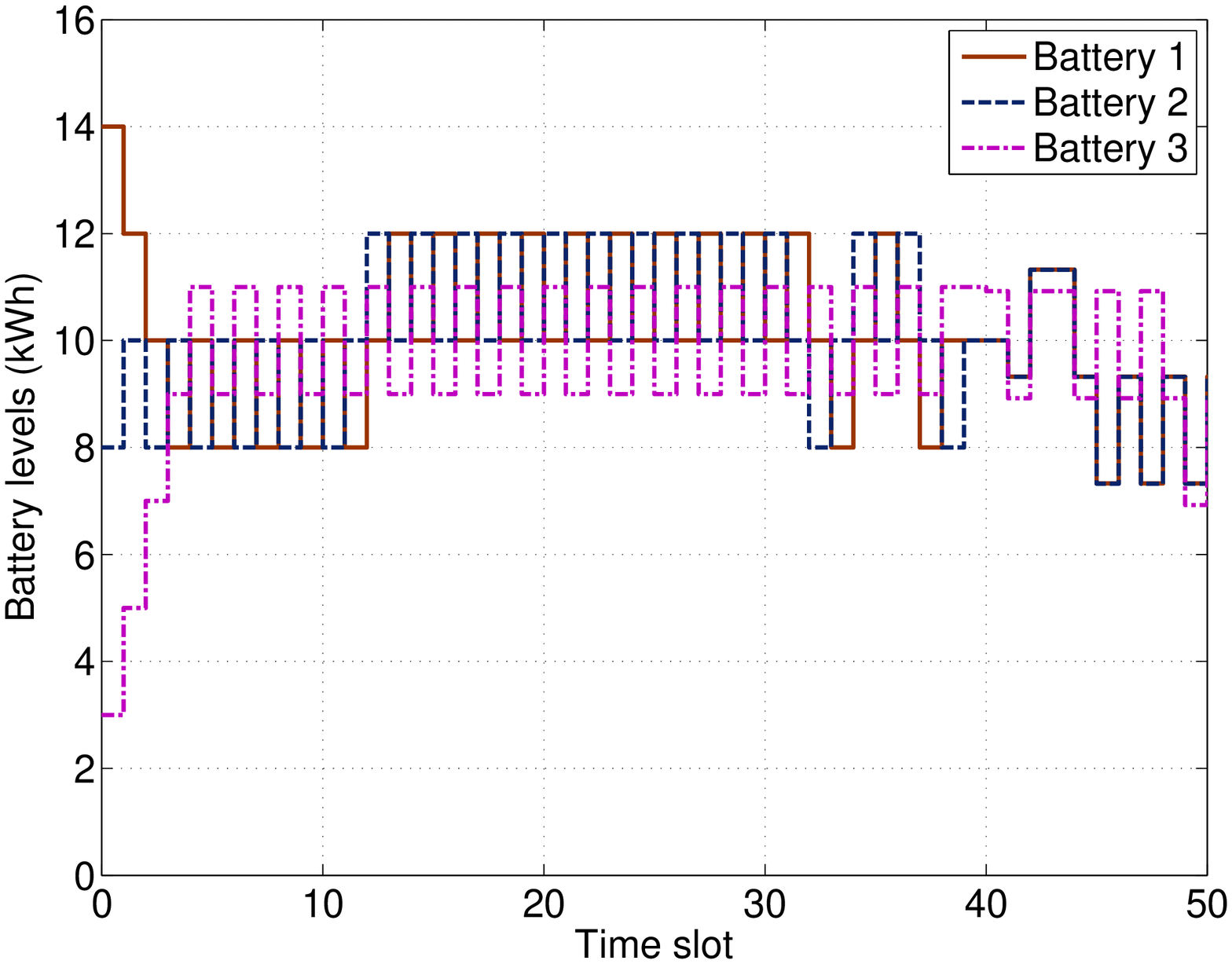}
		\caption{\small Energy levels of three Li-ion batteries ($V = V_{max}$).}
		\label{fig:vmaxbat}
		\vspace{-0.1in}
	\end{minipage}
	\hspace{0.05in}
	\begin{minipage}[t]{.32\linewidth}
		\centering  	
		\includegraphics[width=2.1in, height=1.7in]{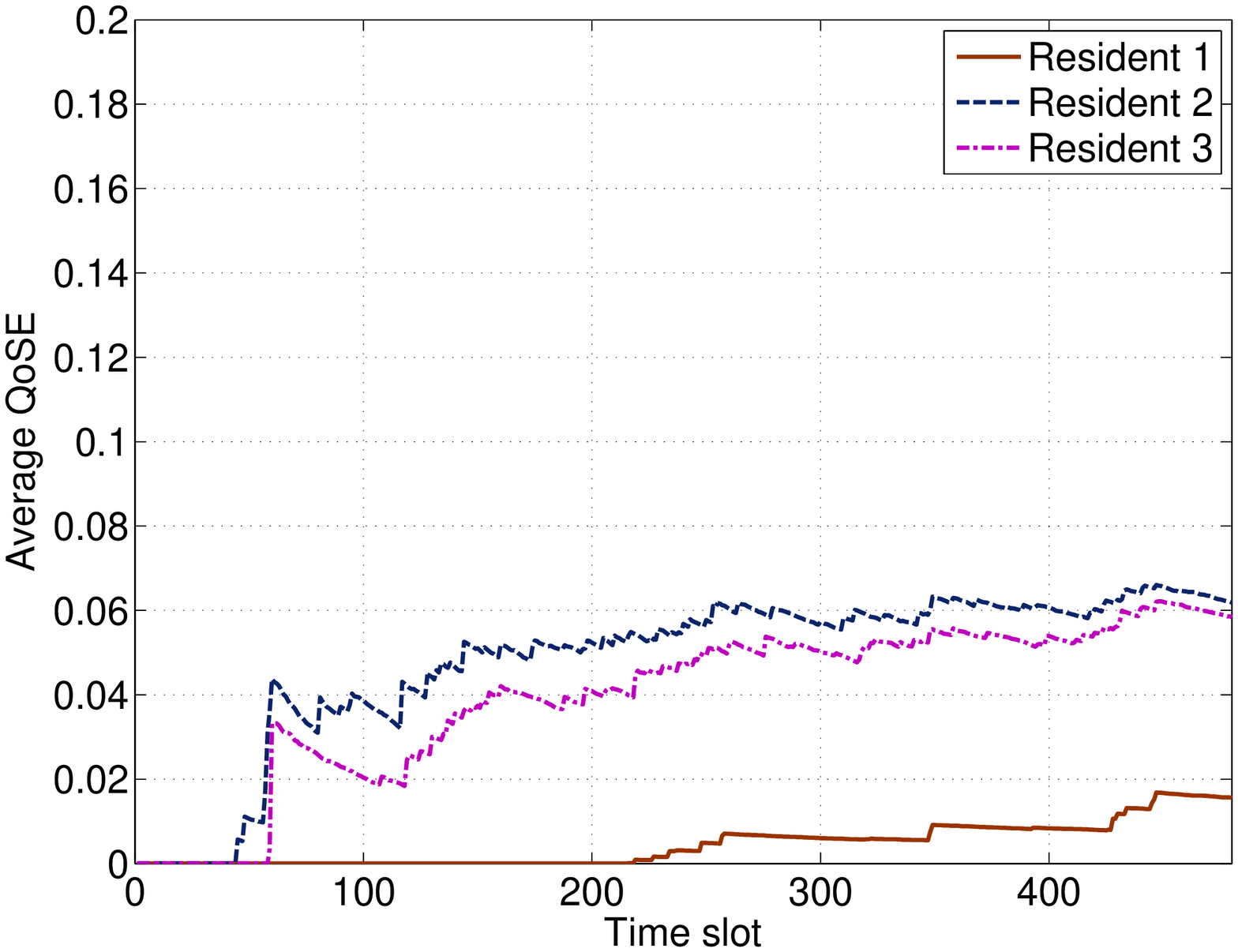} 
		\caption{\small QoSEs for three residents with different service contracts ($V = V_{max}/2$).}
		\vspace{-0.1in}
		\label{fig:hetroqou}
	\end{minipage}	
\end{figure*}

We also plot the MG operation traces from this simulation in Fig.~\ref{fig:vmaxele}. The energy for serving quality usage from the DEREs are plotted in Fig.~\ref{fig:vmaxele}(A). It can be seen that the DRERs generate excessive electricity from slot $150$ to $200$,  which is more than enough for the residents. Thus, the MGCC sells more electricity back to the macrogrid and obtains significant cost compensation accordingly.  In Fig.~\ref{fig:vmaxele}(B), 
we plot the traces of electricity trading, where the positive values are the purchased electricity (marked as brown bars), and the negative values represent the sold electricity (marked as dark blue bars). 
The MG operation costs are plotted in Fig.~\ref{fig:vmaxele}(C). The curve rises when the MG purchases electricity and falls when the MG sells electricity. From slot $150$ to $200$,  the operation cost drops significantly due to profits of selling excess electricity from the DEREs. The operation cost is $\$418.10$ by the end of the period, which means the net spending of the MG is $\$418.10$ on the utility market.   

We then examine the energy levels of the batteries in Fig.~\ref{fig:vmaxbat}. We only plot the levels of three batteries in the first $50$ time slots for clarity. The proposed control policy charges and discharges the batteries in the range of $0$ to $16$ kWh, which falls strictly within the battery capacity limit. It can be seen that the amount of energy for charging or discharging in one slot is limited by $2$ kWh in the figure, due to the short time slots comparing to the $2$-hour fully charge/discharge periods. For longer time slot durations and batteries with faster charge/discharge speeds, the variation of the energy level in Fig.~\ref{fig:vmaxbat} could be higher. However, 
Theorem~\ref{propo:battery} indicates that the feasibility of the battery management constraint is always ensured, if the control parameter $V$ satisfies $0 < V \leq V_{max}$.




We next evaluate the performance of the proposed adaptive control algorithm under different values of control parameter $V$. 
For different values $V =\{ V_{max}, V_{max}/2, V_{max}/4\}$, the QoSEs are stabilized at $0.081$, $0.061$, and $0.055$,
and the total operation cost are \$$418.10$, \$$625.69$, and \$$717.75$, respectively. 
We find the QoSE decreases from $0.081$ to $0.055$, while the total operation cost is increased from \$$418.10$ to \$$717.75$, as $V_{max}$ is decreased. 
This demonstrates the performance-congestion trade-off 
as in Theorem~\ref{propo:objbound}: a larger $V$ leads to a smaller objective value (i.e., the operating cost), but the system is also penalized by a larger virtual queue backlog, which corresponds to a higher QoSE. On the contrary, a smaller $V$ favors the resident quality usage, but increases the total operation cost. In practice, we can select a proper value for this parameter based on the MG design specifications.






It would be interesting to examine the case where the residents require different QoSEs. We assume $5$ residents with a service contract 
for lower QoSEs. We plot the average QoSEs of three residents with $V = V_{max}/2$ in Fig.~\ref{fig:hetroqou}. Resident $1$ prefers an outage probability $\delta_1=0.02$, while residents $2$ and $3$ require an outage probability $\delta_2=\delta_3=0.07$. It can be seen in Fig.~\ref{fig:hetroqou} that resident $1$'s QoSE converges to $0.015$, while the other two residents' QoSEs remains around $0.063$.


\subsection{Comparison with a Benchmark} 

We compare the performance of the proposed scheme with a heuristic MG electricity control policy (MECP), which serves as a benchmark. 
In MECP, the MGCC blocks quality usage requests simply by tossing a coin with the target probability. We use $\delta_n=0.03$ in the following simulations. If there is sufficient electricity from the DRERs, 
all the quality usage requests will be granted and the excess energy will be stored in the ESS's. If there is still any surplus energy, the MGCC will sell it to the macrogrid. If there is insufficient electricity from the DRERs, 
the ESS's will be discharged to serve the quality usage requests. The MGCC will purchase electricity from the macrogrid if even more electricity is required. 
Finally, with a predefined probability, e.g., $0.5$ in the following simulation, the MG purchases as much energy as possible to charge the ESS's.

We run $100$ simulations with different random seeds for a seven-day period. We assume in the first five days the resident behavior is the same as previous default settings. In the last two days, we assume the residents are apt to request more electricity (e.g., more activities in weekends) 
We assume in the last two days the resident pre-agreed basic usage power demand is uniformly distributed from $5$ kW to $35$ kW. The quality usage power is uniformly distributed from $0$ to $20$ kW. 

We find that the proposed algorithm earns 
$\$947.27$ from the utility market (with $95\%$ confidence interval $[950.65,943.89]$). The profit mainly comes from the abundant DRER generation in the last two days, as shown in Fig.~\ref{fig:costcisingle}. MECP only earns $\$379.74$ from the market (with $95\%$ confidence interval $[387.96, 371.52]$), which is $60\%$ lower than that of the proposed control policy. We also find that the QoSEs under the proposed control policy remains about $0.025$, which is lower than the criteria $\delta_n=0.03$. This is because there are a sudden price jump from $\$27$/MWh to $\$356$/MWh in the afternoon of the last day. This sharp increment increases $C_{max}$ eight times and decreases the value of $V_{max}$. Due to the performance-congestion trade-off, the QoSEs become smaller (lower than MECP's $0.03$ level). 



\begin{figure} [!t]
\centering
\includegraphics[width=3.4in, height=3.8in]{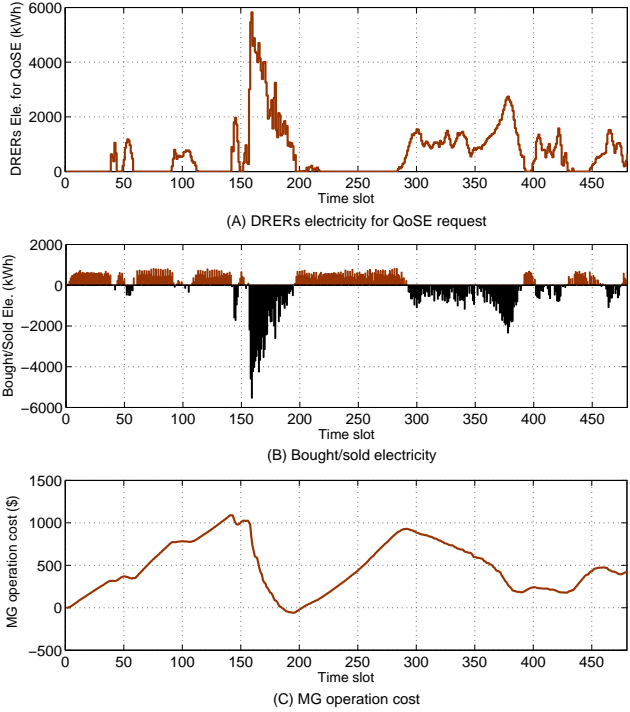}
\caption{MG operation traces of the proposed algorithm for the $5$-day period.} 
\label{fig:vmaxele}
\vspace{-0.15in}
\end{figure}

\begin{figure} [!t]
\centering
\includegraphics[width=3.4in, height=3.8in]{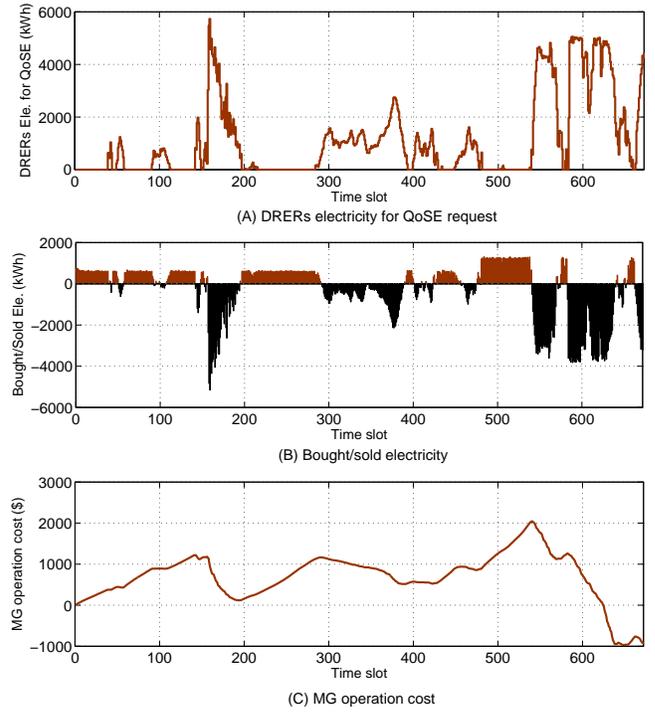}
\caption{MG operation traces of proposed algorithm for the $7$-day period.} 
\label{fig:costcisingle}
\vspace{-0.15in}
\end{figure}

\section{Related Work} \label{sec:related}

SG is regarded as the next generation power grid with two-way flows of electricity and information. 
Several comprehensive reviews of SG technologies can be found in~\cite{Farhangi10, Fang11}.
%
Recently, SG research is attracting considerable interest from the networking and communications communities~\cite{He11,Karimi12,Lu12,Rahman12,Ma12,Liang12}.
For example, the design of wireless communication systems in SG is studied in~\cite{Karimi12}. The authors of~\cite{Lu12,Rahman12} 
explore the important wireless communication security issues in smart grid. The energy management and power flow control in the grid is investigated in~\cite{He11} 
to reach system-wide reliability under uncertainties. 
The frequency oscillation in power networks is studied in~\cite{Ma12} by epidemic propagation and a social network based approach. The electric power management with PHEVs 
are examined in~\cite{Liang12}.

Microgrid is a new grid structure to group DRERs and local residents loads, which provides a promising way for the future SG. In~\cite{Huang08}, the authors review the MG structure with distributed energy resources. In~\cite{Liu10}, 
the integration of random wind power generation into grids for cost effective operation is investigated. In~\cite{Fang112}, the authors propose a useful online method to discover all available DRERs within the islanded mode mircogrid and compute a DRER access strategy. 
The problem of optimal residential demand management is studied in~\cite{Li11},
aiming to adapt to time-varying energy generation and prices, and maximize user benefit. 
In~\cite{Koutsopoulos11}, the authors investigate energy storage management with a dynamic programming approach. 
The size of the ESS's for MG energy storage is explored in~\cite{Chen12}.

Lyapunov optimization is a useful stochastic optimization method~\cite{Tassiulas92}. 
It integrates the Lyapunov stability concept of control theory with optimization and provides an efficient framework for solving schedule and control problems. 
It has been widely used and extended in the communications and networking areas~\cite{Tassiulas92, Neely05}. 
In two recent work~\cite{Neely11, Urgaonkar11}, the Lyapunov optimization method is applied to jointly optimize power procurement and dynamic pricing. In~\cite{Neely11}, the authors investigate the problem of profit maximization for delay tolerant consumers. In~\cite{Urgaonkar11}, the authors study electricity storage management for data centers, aiming to meet the workload requirement. Both of the work are designed based on a \textit{single} energy consumption entity model. 

\section{Conclusion} \label{sec:con}

In this paper, we developed an online adaptive electricity scheduling algorithm for smart energy management in MGs by jointly considering renewable energy penertration, ESS management, residential demand management, and utility market participation. We introduced a QoSE model by taking into account minimization of the MG operation cost, while maintaining the outage probabilities of resident quality usage. We transformed the QoSE control problem and ESS management problem into queue stability problems by introducing the QoSE virtual queues and battery virtual queues. 
The Lyapunov optimization method was applied to solve the problem with an efficient online electricity scheduling algorithm, which has deterministic performance bounds. Our simulation study validated the superior performance of the proposed approach. 

\section*{Acknowledgment}

This work is supported in part by the US National Science Foundation (NSF) under Grants CNS-0953513. 
Any opinions, findings, and conclusions or recommendations expressed in this material are those of the author(s) and do not necessarily reflect the views of the foundation.


\begin{thebibliography}{10}
\providecommand{\url}[1]{#1}
\csname url@samestyle\endcsname
\providecommand{\newblock}{\relax}
\providecommand{\bibinfo}[2]{#2}
\providecommand{\BIBentrySTDinterwordspacing}{\spaceskip=0pt\relax}
\providecommand{\BIBentryALTinterwordstretchfactor}{4}
\providecommand{\BIBentryALTinterwordspacing}{\spaceskip=\fontdimen2\font plus
\BIBentryALTinterwordstretchfactor\fontdimen3\font minus
  \fontdimen4\font\relax}
\providecommand{\BIBforeignlanguage}[2]{{%
\expandafter\ifx\csname l@#1\endcsname\relax
\typeout{** WARNING: IEEEtran.bst: No hyphenation pattern has been}%
\typeout{** loaded for the language `#1'. Using the pattern for}%
\typeout{** the default language instead.}%
\else
\language=\csname l@#1\endcsname
\fi
#2}}
\providecommand{\BIBdecl}{\relax}
\BIBdecl

\bibitem{Huang13}
Y.~Huang, S.~Mao, and R.~M. Nelms, ``Adaptive electricity scheduling in
  microgrids,'' in \emph{Proc. {IEEE} {INFOCOM}'13}, Turin, Italy, Apr. 2013,
  pp. 1--9.

\bibitem{RRA09}
Whitehouse.gov, ``Battery and electric vehicle report,'' Jul. 2010, [online]
  Available:
  \url{http://www.whitehouse.gov/files/documents/Battery-and-Electric-Vehicle-Report-FINAL.pdf}.

\bibitem{Fang11}
X.~Fang, S.~Misra, G.~Xue, and D.~Yang, ``Smart grid - the new and improved
  power grid: A survey,'' \emph{{IEEE} Commun. Surveys \& Tutorials}, vol.~PP,
  no.~99, pp. 1--37, Dec. 2011.

\bibitem{Huang08}
J.~Huang, C.~Jiang, and R.~Xu, ``A review on distributed energy resources and
  microgrid,'' \emph{{ELSEVIER} Renewable and Sustainable Energy Reviews},
  vol.~12, no.~9, pp. 2472--2483, Dec. 2008.

\bibitem{Farhangi10}
H.~Farhangi, ``The path of the smart grid,'' \emph{{IEEE}Power and Energy
  Magazine}, vol.~8, no.~1, pp. 18--28, Jan.-Feb. 2010.

\bibitem{Shao11}
S.~Shao, M.~Pipattanasomporn, and S.~Rahman, ``Demand response as a load
  shaping tool in an intelligent grid with electric vehicles,'' \emph{{IEEE}
  Trans. Smart Grid}, vol.~2, no.~4, pp. 624--631, Dec. 2011.

\bibitem{Fumagalli04}
E.~Fumagalli, J.~W. Black, I.~Vogelsang, and M.~Ilic, ``Quality of service
  provision in electric power distribution systems through reliability
  insurance,'' \emph{{IEEE} Trans. Power Systems}, vol.~19, no.~3, pp.
  1286--1293, Aug. 2004.

\bibitem{Tassiulas92}
L.~Tassiulas and A.~Ephremides, ``Stability properties of constrained queueing
  systems and scheduling policies for maximum throughput in multihop radio
  networks,'' \emph{{IEEE} Trans. Autom. Control}, vol.~37, no.~12, pp.
  1936--1948, Dec. 1992.

\bibitem{Slotine91}
J.~Slotine and W.~Li, \emph{{A}pplied {N}onlinear {C}ontrol}.\hskip 1em plus
  0.5em minus 0.4em\relax Prentice Hall, 1991.

\bibitem{Kim11}
T.~T. Kim and H.~Poor, ``Scheduling power consumption with price uncertainty,''
  \emph{{IEEE} Trans. Smart Grid}, vol.~2, no.~3, pp. 519--527, Sept. 2011.

\bibitem{Neely05}
M.~Neely, E.~Modiano, and C.~Rohrs, ``Dynamic power allocation and routing for
  time-varying wireless networks,'' \emph{{IEEE} J. Sel. Areas Commun.},
  vol.~23, no.~1, pp. 89--103, Jan. 2005.

\bibitem{Neely08}
M.~J. Neely and R.~Urgaonkar, ``Opportunism, backpressure, and stochastic
  optimization with the wireless broadcast advantage,'' in \emph{Asilomar
  Conference on Signals, Systems, and Computers'08}, Pacific Grove, CA, Oct.
  2008, pp. 1--7.

\bibitem{WWR}
The National Renewable Energy Laboratory, ``Western wind resources dataset,''
  [online] Available: \url{http://wind.nrel.gov/Web_nrel/}.

\bibitem{Peterson10}
S.~B. Peterson, J.~F. Whitacre, and J.~Apt, ``The economics of using plug-in
  hybrid electric vehicle battery packs for grid storage,'' \emph{J. Power
  Sources}, vol. 195, no.~8, pp. 2377--2384, 2010.

\bibitem{ERCOT}
``The Electric Reliability Council of Texas,'' [online] Available:
  \url{http://www.ercot.com/}.

\bibitem{He11}
M.~He, S.~Murugesan, and J.~Zhang, ``Multiple timescale dispatch and scheduling
  for stochastic reliability in smart grids with wind generation integration,''
  in \emph{Proc. {IEEE} {INFOCOM}'11}, Shanghai, China, Apr. 2011, pp.
  461--465.

\bibitem{Karimi12}
B.~Karimi and V.~Namboodiri, ``Capacity analysis of a wireless backhaul for
  metering in the smart grid,'' in \emph{Proc. {IEEE} {INFOCOM}'12}, Orlando,
  FL, Mar. 2012, pp. 61--66.

\bibitem{Lu12}
Z.~Lu, W.~Wang, and C.~Wang, ``Hiding traffic with camouflage: Minimizing
  message delay in the smart grid under jamming,'' in \emph{Proc. {IEEE}
  {INFOCOM}'12}, Orlando, FL, Mar. 2012, pp. 3066--3070.

\bibitem{Rahman12}
M.~A. Rahman, P.~Bera, and E.~Al-Shaer, ``Smartanalyzer: A noninvasive security
  threat analyzer for {AMI} smart grid,'' in \emph{Proc. {IEEE} {INFOCOM}'12},
  Orlando, FL, Mar. 2012, pp. 2255--2263.

\bibitem{Ma12}
H.~Ma, H.~Li, and Z.~Han, ``A framework of frequency oscillation in power grid:
  Epidemic propagation over social networks,'' in \emph{Proc. {IEEE}
  {INFOCOM}'12}, Orlando, FL, Mar. 2012, pp. 67--72.

\bibitem{Liang12}
H.~Liang, B.~J. Choi, W.~Zhuang, and X.~Shen, ``Towards optimal energy
  store-carry-and-deliver for {PHEV}s via {V2G} system,'' in \emph{Proc. {IEEE}
  {INFOCOM}'12}, Orlando, FL, Mar. 2012, pp. 1674--1682.

\bibitem{Liu10}
X.~Liu, ``Economic load dispatch constrained by wind power availability: A
  wait-and-see approach,'' \emph{{IEEE} Smart Grid}, vol.~1, no.~3, pp.
  347--355, Dec. 2010.

\bibitem{Fang112}
X.~Fang, D.~Yang, and G.~Xue, ``Online strategizing distributed renewable
  energy resource access in islanded microgrids,'' in \emph{{IEEE}
  {GLOBECOM}'11}, Huston, TX, Dec. 2011, pp. 1931--1937.

\bibitem{Li11}
N.~Li, L.~Chen, and S.~H. Low, ``Optimal demand response based on utility
  maximization in power networks,'' in \emph{2011 {IEEE} {PES General
  Meeting}}, Detroit, MI, Jul. 2011, pp. 1--8.

\bibitem{Koutsopoulos11}
I.~Koutsopoulos, V.~Hatzi, and L.~Tassiulas, ``Optimal energy storage control
  policies for the smart power grid,'' in \emph{{IEEE} {SmartGridComm}'11},
  Oct. 2011, pp. 475--480.

\bibitem{Chen12}
S.~X. Chen, H.~B. Gooi, and M.~Q. Wang, ``Sizing of energy storage for
  microgrids,'' \emph{{IEEE} Trans. Smart Grid}, vol.~3, no.~1, pp. 142--151,
  Mar. 2012.

\bibitem{Neely11}
M.~J. Neely, A.~S. Tehrani, and A.~G. Dimakis, ``Efficient algorithms for
  renewable energy allocation to delay tolerant consumers,'' in \emph{{IEEE}
  {SmartGridComm}'10}, Oct. 2010, pp. 549--554.

\bibitem{Urgaonkar11}
R.~Urgaonkar, B.~Urgaonkar, M.~J. Neely, and A.~Sivasubramaniam, ``Optimal
  power cost management using stroed energy in data centers,'' in \emph{Proc.
  {ACM} {SIGMETRICS}'11}, San Jose, CA, Jun. 2011, pp. 221--232.

\end{thebibliography}



\appendices

\section{Derivation of Equation~(\ref{eq:drift})} \label{app:drift}

With the drift defined as in (\ref{eq:driftdef}), we have
\begin{eqnarray}
\Delta(\vec{\Theta}(t)) &=& \frac{1}{2} \mathbb{E} \left\{ \sum_{k=1}^K [(X_k(t+1))^2 - (X_k(t))^2 | X_k(t)]+ \right. \nonumber\\
             & & \hspace*{-0.6in} \left. \sum_{n=1}^N [(Z_n(t+1))^2 - (Z_n(t))^2 | Z_n(t)] \right\}\nonumber
\end{eqnarray}
\begin{eqnarray}
            &\hspace*{-1.0in} \leq& \hspace*{-0.6in} \frac{1}{2}\mathbb{E} \left\{ \sum_{k=1}^K \left[ (D_k(t))^2 + (R_k(t))^2 + 2X_k(t)(R_k(t) - \right.\right. \nonumber\\
             & & \hspace*{-0.6in} \left. \left. D_k(t))|X_k(t) \right] \right\} + \frac{1}{2} \mathbb{E} \left\{\sum_{n=1}^N \left[ I_n(t)^2+ \right. \right. \nonumber \\
             & & \hspace*{-0.6in} (\delta_n\alpha_n(t))^2 + 2Z_n(t)(I_n(t) - \delta_n\alpha_n(t))|Z_n(t)\}\nonumber
\end{eqnarray}
\begin{eqnarray}
             &\hspace*{-1.0in}=& 
             \hspace*{-0.6in} \frac{1}{2} \sum_{k=1}^K \mathbb{E}\{[(D_k(t))^2+(R_k(t))^2]\} + \nonumber \\
             & & \hspace*{-0.6in} \sum_{k=1}^K \mathbb{E}\{X_k(t)(R_k(t) - D_k(t))|X_k(t)\} + \nonumber 
\end{eqnarray}
\begin{eqnarray}
             & & \hspace*{-0.6in} 
             \frac{1}{2} \sum_{n=1}^N \mathbb{E}\{[(1+(\sigma_n)^2)(\alpha_n(t))^2 + (p_n(t))^2]\} +  \nonumber\\
             & & \hspace*{-0.6in} \frac{1}{2} \sum_{n=1}^N\mathbb{E}\{2Z_n(t)(1-\delta_n(t))\alpha_n(t) - \nonumber\\
             & & \hspace*{-0.6in} (Z_n(t) + \alpha_n(t))p_n(t)|Z_n(t)\}\nonumber\\
             &\hspace*{-1.0in}\leq& \hspace*{-0.6in} B+\sum_{n=1}^N\mathbb{E}\{Z_n(t)(1-\delta_n)\alpha_n(t)|Z_n(t)\} + \nonumber\\
             & & \hspace*{-0.6in} \sum_{k=1}^K\mathbb{E}\{X_k(t)(R_k(t) - D_k(t))|X_k(t)\} - \nonumber \\
             & & \hspace*{-0.6in} \sum_{n=1}^N\mathbb{E}\{(Z_n(t) + \alpha_n(t))p_n(t)|Z_n(t)\}. \nonumber
\end{eqnarray}
where $B = \frac{1}{2} \sum_{k=1}^K (\max\{D_k^{max}, R_k^{max}\})^2 + \frac{1}{2} \sum_{n=1}^N (2+\delta_n^2)(\alpha_n^{max})^2$ is a constant.

\section{Proof of Lemma~\ref{lm:lm1}} \label{app:lm1}

\begin{proof}
In part 1) of Lemma~\ref{lm:lm1}, if $Q(t) > 0$, we have $S(t) = 0$ according to~(\ref{eq:sellconst}). 
The the objective function of problem~(\ref{eq:objdrift}) becomes
\begin{eqnarray} \label{eq:lmqt}
&& \hspace{-0.7in}  VQ(t)C(t) + \sum_{k=1}^K{X_k(t)(R_k(t) - D_k(t))} -  \nonumber\\ 
&& \hspace{0.7in} \sum_{n=1}^N{(Z_n(t) + \alpha_n(t))p_n(t)}.
\end{eqnarray}

We first prove Lemma~\ref{lm:lm1}-\ref{item:item1a}). If $X_k(t) > -VC(t)$, we assume $R_k(t) > 0$. Then we have $D_k(t) = 0$ according to~(\ref{eq:batredecontr}). Accordingly, the object function~(\ref{eq:lmqt}) is transformed to
\begin{eqnarray}
&&VQ(t)C(t) + \sum_{i \neq k}{X_i(t)(R_i(t) - D_i(t))} - \nonumber\\ 
&&\sum_{n=1}^N{(Z_n(t) + \alpha_n(t))p_n(t)} + X_k(t)R_k(t) \nonumber
\end{eqnarray}
\begin{eqnarray}
&>&VQ(t)C(t) + \sum_{i \neq k}{X_i(t)(R_i(t) - D_i(t))} - \nonumber\\ 
&&\sum_{n=1}^N{(Z_n(t) + \alpha_n(t))p_n(t)} - VC(t)(P(t)+Q(t)- \nonumber\\
&&\sum_{i \neq k}(R_i(t) - D_i(t)) - \sum_{n=1}^Np_n(t) \nonumber\\
&=&V \left[ \sum_{i \neq k}(R_i(t) - D_i(t)) + \sum_{n=1}^Np_n(t) -P(t) \right] C(t)+\nonumber\\
&& \sum_{i \neq k}{X_i(t)(R_i(t) - D_i(t))} - \sum_{n=1}^N{(Z_n(t) + \alpha_n(t))p_n(t)}. \nonumber
\end{eqnarray}
The above inequality is due to $X_k(t) > -VC(t)$ and $R_k(t) = P(t) + Q(t) - \sum_{i \neq k}(R_i(t) - D_i(t)) - \sum_{n=1}^Np_n(t) \geq 0$. The last expression shows, given the assumption $R_k(t) > 0$, we may find another feasible electricity allocation scheme $\tilde{Q}(t) = \sum_{i \neq k}(R_i(t) - D_i(t)) + \sum_{n=1}^Np_n(t)-P(t)$, which can achieve a smaller objective value 
by choosing $R_k(t) = 0$ and $D_k(t) = 0$. 
This contradicts with the assumption $R_k(t) > 0$. Thus, we prove that $R_k(t) = 0$ when $X_k(t) > -VC(t)$, under the situation $Q(t) > 0, S(t) = 0$.

We then prove the second part of Lemma~\ref{lm:lm1}-\ref{item:item1a}). It follows~(\ref{eq:batredecontr}) that $R_k(t) = 0$ if $D_k(t) > 0$. 
Then~(\ref{eq:lmqt}) becomes
\begin{eqnarray}
&&\hspace{-0.2in} VQ(t)C(t) + \sum_{i \neq k}{X_i(t)(R_i(t) - D_i(t))} - \nonumber\\ 
&&\hspace{-0.2in} \sum_{n=1}^N{(Z_n(t) + \alpha_n(t))p_n(t)} - X_k(t)D_k(t) \nonumber\\
&\hspace{-0.4in} >&\hspace{-0.2in} VQ(t)C(t) + \sum_{i \neq k}{X_i(t)(R_i(t) - D_i(t))} - \nonumber\\ 
&&\hspace{-0.2in} \sum_{n=1}^N{(Z_n(t) + \alpha_n(t))p_n(t)} + VC(t)(\sum_{n=1}^Np_n(t) - P(t)- \nonumber\\
&&\hspace{-0.2in} Q(t)+\sum_{i \neq k}(R_i(t) - D_i(t)) ) \nonumber\\
&\hspace{-0.4in} =&\hspace{-0.2in} V \left[ \sum_{i \neq k}(R_i(t) - D_i(t)) + \sum_{n=1}^Np_n(t)-P(t) \right] C(t)+\nonumber\\
&&\hspace{-0.2in} \sum_{i \neq k}{X_i(t)(R_i(t) - D_i(t))} - \sum_{n=1}^N{(Z_n(t) + \alpha_n(t))p_n(t)}.\nonumber
\end{eqnarray}
The above inequality is due to $X_k(t) < -VC(t) < 0$ and $D_i(t) = -P(t) - Q(t) + \sum_{i \neq k}(R_i(t) + D_k(t)) + \sum_{n=1}^Np_n(t) > 0$. The last expression shows, given the assumption $D_k(t) > 0$, we may find another electricity allocation scheme with $\tilde{Q}(t) = \sum_{i \neq k}(R_i(t) - D_i(t)) + \sum_{n=1}^Np_n(t) - P(t)$, which can achieve a smaller objective value by choosing $R_k(t) = 0$ and $D_k(t) = 0$. 
This contradicts with the assumption $D_k(t)>0$. We thus prove that $D_k(t) = 0$ when $X_k(t) < -VC(t)$, under the situation $Q(t) > 0, S(t) = 0$, which completes the proof of Lemma~\ref{lm:lm1}-\ref{item:item1a}).

We next prove Lemma~\ref{lm:lm1}-\ref{item:item1b}). For the first part, if $Z_n(t) > VC(t) - \alpha_n(t)$, we assume $0 \leq p_n(t) < (1-\delta_n)\alpha_n(t)$. Following~(\ref{eq:objdriftfull}) and $S(t) = 0$, we have
\begin{eqnarray}
&&\hspace{-0.15in} B+VQ(t)C(t) + \sum_{k=1}^KX_k(t)(R_k(t) - D_k(t)) + \nonumber\\
&&\hspace{-0.15in} \sum_{j\neq n}(Z_j(t)(1-\delta_j)\alpha_j(t) - (Z_j(t) + \alpha_j(t))p_j(t)) + \nonumber\\
&&\hspace{-0.15in} Z_n(t)(1-\delta_n)\alpha_n(t) - (Z_n(t)+\alpha_n(t))p_n(t) \nonumber
\end{eqnarray}
\begin{eqnarray}
&\hspace{-0.2in} =&\hspace{-0.15in} B+VQ(t)C(t) + \sum_{k=1}^KX_k(t)(R_k(t) - D_k(t)) + \nonumber\\
&&\hspace{-0.15in} \sum_{j\neq n}(Z_j(t)(1-\delta_j)\alpha_j(t) - (Z_j(t) + \alpha_j(t))p_j(t)) + \nonumber\\
&&\hspace{-0.15in} Z_n(t)[(1-\delta_n)\alpha_n(t) - p_n(t)] - \alpha_n(t)p_n(t) \nonumber
\end{eqnarray}
\begin{eqnarray}
&\hspace{-0.2in} >&\hspace{-0.15in}  B+VQ(t)C(t) + \sum_{k=1}^KX_k(t)(R_k(t) - D_k(t)) + \nonumber
\end{eqnarray}
\begin{eqnarray}
&&\hspace{-0.15in} \sum_{j\neq n}(Z_j(t)(1-\delta_j)\alpha_j(t) - (Z_j(t) + \alpha_j(t))p_j(t)) + \nonumber\\
&&\hspace{-0.15in} (VC(t) - \alpha_n(t))[(1-\delta_n)\alpha_n(t) - p_n(t)] - \alpha_n(t)p_n(t) \nonumber
\end{eqnarray}
\begin{eqnarray}
&\hspace{-0.2in} =&\hspace{-0.15in}  B+V[\sum_{k=1}^K(R_k(t) - D_k(t)) + \sum_{j\neq n}p_j(t) -P(t) + \nonumber\\
&&\hspace{-0.15in} (1-\delta_n)\alpha_n(t)]C(t) + \sum_{k=1}^KX_k(t)(R_k(t) - D_k(t)) + \nonumber\\
&&\hspace{-0.15in} \sum_{j\neq n}(Z_j(t)(1-\delta_j)\alpha_j(t) - (Z_j(t) + \alpha_j(t))p_j(t)) + \nonumber\\
&&\hspace{-0.15in}  Z_n(t)(1-\delta_n)\alpha_n(t) - (Z_n(t) + \alpha_n(t))(1-\delta_n)\alpha_n(t).\nonumber
\end{eqnarray}
The above inequality is due to $Z_n(t) > VC(t) - \alpha_n(t)$ and the assumption $p_n(t) < (1-\delta_n)\alpha_n(t)$. The last equality shows, given the assumption $p_n(t) < (1-\delta_n)\alpha_n(t)$, we may find another electricity allocation scheme with $p_n(t) = (1-\delta_n)\alpha_n(t)$ and $\tilde{Q}(t) = \sum_{k=1}^K(R_k(t) - D_k(t)) + \sum_{j\neq n}p_j(t) -P(t) + (1-\delta_n)\alpha_n(t)$, which can achieve a smaller objective value. This contradicts with the previous assumption. Thus, we have $p_n(t) \geq (1-\delta_n)\alpha_n(t)$.

For the second part of Lemma~\ref{lm:lm1}-\ref{item:item1b}), assume $p_n(t) > 0$ for $0 \leq Z_n(t) < VC(t) - \alpha_n(t)$.It follows~(\ref{eq:sellconst}) that $S(t) = 0$. The objective function~(\ref{eq:objdriftfull}) can be written as
\begin{eqnarray}
&&B+VQ(t)C(t) + \sum_{k=1}^KX_k(t)(R_k(t) - D_k(t)) + \nonumber\\
&&\sum_{j\neq n}(Z_j(t)(1-\delta_j)\alpha_j(t) - (Z_j(t) + \alpha_j(t))p_j(t)) + \nonumber\\
&&Z_n(t)(1-\delta_n)\alpha_n(t) - (Z_n(t)+\alpha_n(t))p_n(t) \nonumber 
\end{eqnarray}
\begin{eqnarray}
&>& B+VQ(t)C(t) + \sum_{k=1}^KX_k(t)(R_k(t) - D_k(t)) + \nonumber \\
&&\sum_{j\neq n}(Z_j(t)(1-\delta_j)\alpha_j(t) - (Z_j(t) + \alpha_j(t))p_j(t)) + \nonumber\\
&&Z_n(t)(1-\delta_n)\alpha_n(t) - VC(t)p_n(t) \nonumber
\end{eqnarray}
\begin{eqnarray}
&\geq& B+V[-P(t)+\sum_{k=1}^K(R_k(t) - D_k(t)) + \nonumber\\
&&\sum_{j\neq n}p_j(t)]C(t) + \sum_{k=1}^KX_k(t)(R_k(t) - D_k(t)) + \nonumber\\
&&\sum_{j\neq n}(Z_j(t)(1-\delta_j)\alpha_j(t)) - \sum_{j\neq n}((Z_j(t) + \alpha_j(t))p_j(t)). \nonumber
\end{eqnarray}
The first inequality is due to $0 \leq Z_n(t) < VC(t) - \alpha_n(t)$ and the assumption $p_n(t) > 0$. The second inequality 
is due to the non-negativity of $Z_n(t)$ and $\alpha_n(t)$. The last equation shows, given the assumption $p_n(t) > 0$, we may find another electricity allocation scheme with $p_n(t) = 0$ and $\tilde{Q}(t) = -P(t)+\sum_{k=1}^K(R_k(t) - D_k(t)) + \sum_{j\neq n}p_j(t)$, which can achieve a smaller objective value. This contradicts with the previous assumption. Thus, we have $p_n(t) =0$, which completes the proof of Lemma~\ref{lm:lm1}-\ref{item:item1b}).

In part~\ref{enu:enu2}) of Lemma~\ref{lm:lm1}, if $S(t) > 0$, we have $Q(t) = 0$ according to~(\ref{eq:sellconst}). The objective function~(\ref{eq:objdrift}) becomes 
\begin{eqnarray} \label{eq:lmwt}
&& \hspace{-0.7in} -VS(t)W(t) + \sum_{k=1}^K{X_k(t)(R_k(t) - D_k(t))} - \nonumber\\ 
&& \hspace{0.7in} \sum_{n=1}^N{(Z_n(t) + \alpha_n(t))p_n(t)}.
\end{eqnarray}
We can prove part~\ref{enu:enu2}) with a similar approach as in the case of part~\ref{enu:enu1}. The detailed proof is omitted for brevity.
\end{proof}

\section{Proof of Lemma~\ref{lm:lm2}} \label{app:lm2}

\begin{proof}
Since $0 \leq C_{min} \leq C(t) \leq C_{max}$ and $V > 0$, we have $R_k(t) = 0$ when $X_k(t) < -VC_{max}$, and $D_k(t) = 0$ when $X_k(t) > -VC_{min}$ according to Lemma~\ref{lm:lm1}-\ref{enu:enu1}). Similarly, since $0 \leq W_{min} \leq W(t) \leq W_{max}$ and $V > 0$, we obtain $R_k(t) = 0$ when $X_k(t) < -VW_{max}$, and $D_k(t) = 0$ when$X_k(t) > -VW_{min}$ 
according to Lemma~\ref{lm:lm1}-\ref{enu:enu2})

Since $C_{max} > W_{max}$ and $C_{min} > W_{min}$,  we conclude that if $X_k(t) > -VW_{min}$, the optimal solution always select $R_k(t) = 0$. If $X_k(t) < -VC_{max}$, the optimal solution always select $D_k(t) = 0$. The proof is completed. 
\end{proof}

\end{document}